\documentclass[aps,prl,notitlepage,twocolumn,superscriptaddress,longbibliography,nofootinbib]{revtex4-1}
\pdfoutput=1
\usepackage[utf8]{inputenc}

\usepackage{amsmath,amsthm,amssymb,amsfonts}
\usepackage{scalerel}
\usepackage{mathtools}
\usepackage{graphicx}
\usepackage{stfloats}
\usepackage{subfigure}
\usepackage[normalem]{ulem}
\usepackage[colorlinks = true,
            linkcolor = blue,
            urlcolor  = magenta,
            citecolor = blue,
            anchorcolor = blue]{hyperref}
\usepackage{mathrsfs}
\usepackage{bbold}
\usepackage{units}
\allowdisplaybreaks
\usepackage{bm}
\usepackage{braket}
\usepackage{makecell}
\usepackage[nameinlink]{cleveref} 
\usepackage{upgreek}
\usepackage{blindtext}
\usepackage{verbatim}
\usepackage{algorithm}
\usepackage[noend]{algpseudocode}
\usepackage[dvipsnames]{xcolor}
\usepackage{bbm}
\usepackage{array}
\usepackage{multirow}
\usepackage{tabularx}
\usepackage{float}
\usepackage{dcolumn}
\usepackage{slashed}
\usepackage{braket}
\usepackage{verbatim}
\usepackage{multirow}
\usepackage{resizegather}
\usepackage{titlesec}
\usepackage[toc,page]{appendix}
\usepackage[export]{adjustbox}

\newcommand{\tr}{\mathrm{tr}}
\theoremstyle{definition}
\theoremstyle{plain}
\newtheorem{prop}{Property}
\theoremstyle{plain}
\newtheorem*{prop*}{Property}
\newtheorem{claim}{Claim}
\newtheorem{lemma}{Lemma}

\begin{document}

\title{{Spread of Entanglement in Generalized Kicked Ising Chain\\
}}
\author{Tanay Pathak\,\,\href{https://orcid.org/0000-0003-0419-2583}
{\includegraphics[scale=0.05]{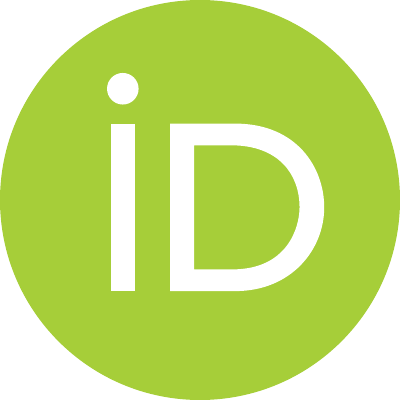}}}
\email{pathak.tanay.4s@kyoto-u.ac.jp}
\affiliation{Department of Physics, Kyoto University, Kitashirakawa Oiwakecho, Sakyo-ku, Kyoto 606-8502, Japan}

\author{Hiromi Ebisu\,\,\href{https://orcid.org/0000-0001-8856-8562}
{\includegraphics[scale=0.05]{orcidid.pdf}}}
\email{hiromi.phys@gmail.com}
\affiliation{
Interdisciplinary Theoretical and Mathematical Sciences Program (iTHEMS), RIKEN,
Wako 351-0198, Japan}

\author{Toma\v{z} Prosen}
\email{tomaz.prosen@fmf.uni-lj.si}
\affiliation{
Department of Physics, Faculty of Mathematics and Physics,
University of Ljubljana, Jadranska 19, SI-1000 Ljubljana, Slovenia}

\begin{abstract}
We investigate the dynamics of entanglement in a generalized version of the kicked Ising chain, extending the model from the standard qubit case (local dimension $q=2$) to higher local dimensions ($q > 2$). We identify the existence of ``dual-unitary'' points where the model's space-time duality allows for exact analytical solutions. Our analysis reveals that while a few unique dual-unitary points exist analytically for systems with local dimensions $q=3$ and $q=4$, such points do not exist for $q \ge 5$ due to the lack of a unique kicking strength that satisfies the required matrix element conditions. Utilizing the transfer matrix method and a replica trick specifically adapted for higher dimensions, we derive exact expressions for the growth of entanglement entropy in the $q=3$ (kicked Potts-type) model starting from a class of solvable initial states. Our results demonstrate that at the dual-unitary point, both R\'enyi and von Neumann entanglement entropies grow linearly with time until reaching a maximum value determined by the subsystem size.

\end{abstract}
\maketitle
\begingroup
\renewcommand{\thefootnote}{}
\footnotetext{The authors are listed in order of their relative contributions.}
\addtocounter{footnote}{-1}
\endgroup

\section{Introduction}
The study of entanglement dynamics in quantum many-body systems has become a central theme in modern theoretical physics, particularly for understanding the thermalization of isolated quantum systems and the characterization of quantum chaos. Among the various models used to explore these phenomena, the kicked Ising model (KIM) \cite{prosenkfim,PhysRevE.65.036208,Prosen:2007hwp} has served as a paradigmatic example, which is known to possess a ``dual-unitary point'', a particle time symmetry in the model, at a special point \cite{Akila_2016} (also see \cite{Bertini:2025ddr} for a comprehensive review). This duality property can then be generalized to dual-unitary circuits allowing for the analytical calculation of various measures of quantum chaos such as the correlation functions \cite{PhysRevLett.123.210601,PhysRevB.102.174307,PhysRevLett.126.100603,PhysRevB.101.094304,PhysRevLett.133.170403}, the growth of entanglement \cite{Bertini:2018fbz,PhysRevLett.132.120402,Gopalakrishnan:2019pip,Bertini:2019gbu,Bertini:2019wkb,PhysRevLett.125.070501,Reid:2021fsg,Zhou:2022uuw,PhysRevB.107.174311,PRXQuantum.6.010324,Pathak:2026rfk,Pathak:2026jfn}, the spectral form factor \cite{Bertini:2018wlu,Bertini:2020mdd,PhysRevResearch.2.043403,PhysRevX.11.021051,PhysRevResearch.6.033226,PhysRevB.111.094316}.

In this work, we 
extend the study of the kicked Ising model to systems with higher local dimensions, specifically where $q > 2$. This generalization allows us to investigate how the properties of entanglement spread and the existence of dual-unitarity evolve as the Hilbert space of the individual constituents becomes more complex. In particular, we focus on the generalized kicked Ising chain, which for $q=3$ corresponds to a kicked Potts-type model (KPM)\cite{PhysRevB.105.144306,Claeys:2024tuy}. A key objective of this work is to determine the conditions under which these generalized models exhibit dual-unitarity. While dual-unitarity is a robust feature for $q=2$, our analysis shows that for higher dimensions, such a point exists analytically only for $q \le 4$. 
To wit, for $q \ge 5$, we find that there is no unique value of the kick strength $b$ that satisfies the requirements for dual-unitarity across all matrix elements.\par
Beyond establishing the existence and breakdown of dual-unitary points, we further investigate the resulting entanglement dynamics in the solvable $q=3$ regime. We identify a class of analytically tractable initial states for which the entanglement evolution can be computed exactly using a transfer-matrix approach and replica techniques adapted to higher local dimensions. At the dual-unitary point, we show that both Rényi and von Neumann entropies exhibit ballistic growth with maximal entanglement velocity until saturation. We further complement these exact results with numerical studies of generic initial states, revealing that signatures of near-maximal entanglement spreading persist even beyond the exactly solvable setting.\par
The contribution of this work is twofold. First, it establishes the limits of dual-unitarity in generalized Ising models, demonstrating that while solvability persists for $q=3$ and $q=4$, it encounters fundamental constraints as the local dimension increases further. Second, by deriving exact analytical expressions for the entanglement growth in these chains with larger local dimensions provides a new benchmark for understanding how quantum information spreads in non-qubit systems. These results are of further importance to the current active area of quantum computing architecture. It is known from previous results \cite{PhysRevA.101.022304} that to maximize the improvement of quantum computer, it is important to use qudits with a certain number of internal states depending on the spatial
topology and connectivity of a quantum system. The Potts models furthermore is special because it admits a description using parafermions, type of particles which obey non-trivial quasi-local
anyonic statistics, which are also linked to topological quantum
computing \cite{PhysRevX.5.041040,Alicea:2015hja,PhysRevB.93.125105,trebst2008short}. 

The rest of this paper is organized as follows. In Sec.~\ref{sec2}, we introduce the generalized KIM. In Sec.~\ref{sec:dualexist}, we investigate the condition to obtain the dual-unitarity point for the model. In Sec.~\ref{sec4} we setup the problem. In Sec.~\ref{sec5} we state various properties of the transfer matrix involved. In Sec.~\ref{sec6} we provide various numerical results for solvable states, generic states, integrable and weak integrability breaking systems at finite systems size and results in the thermodynamic  limit using the duality mapping. 
\section{Model}\label{sec2}
We first recall the kicked KFIM model for which the Hamiltonian is given by 

\begin{equation}\label{eq:hamxyz}
    H= H_{I} + H_{K} \sum_{\tau= -\infty}^{\infty} \delta(t- \tau)
\end{equation}
where 
\begin{align}
    H_{I}& = \sum_{i=1}^{N} J \sigma^{z}_{i}\sigma^{z}_{i+1} + \sum_{i}h_{i}\sigma_{i}^{z}, \quad
    H_{K}& = \sum_{i=1}^{N} b\, \sigma^{x}_{i}.
\end{align}
The total Floquet operator of the system is 
\begin{equation}
    U_{KI}= U_{K}U_{I},
\end{equation}
where $U_{I}= e^{-i H_{I}}$ and $U_{K}= e^{-i H_{K}}$. Also, $h_{i}$ can be generic. 

Here we generalize the Ising type model, which has local dimension $q=2$, to models, which have local dimension $q > 2$. For simplicity, we first focus on the case $q=3$. In this case, the model corresponds to kicked Potts model~\cite{PhysRevB.105.144306}. Hamiltonian is given by

\begin{align}\label{eqn:ising3}
 H_{I}& = J\sum_{i=1}^{N} {Z}_{i}{Z}^{\dagger}_{i+1}+ \sum_{i}h_{i}Z_{i}+h.c.\nonumber \\
    H_{K}& = b\sum_{i=1}^{N} \, X_{i}+h.c.,
\end{align}
where $h_{i}$ 
denotes
generic Gaussian random numbers. Also, ``h.c." stands for Hermitian conjugate. 
Defining $\omega= e^{2 \pi i/3}$, we can write the explicit representation of $Z$ and $X$ as 
\begin{equation}
Z=\left(\begin{array}{ccc}
1 & 0 & 0 \\
0 & \omega & 0 \\
0 & 0 & \omega^2
\end{array}\right), \quad X=\left(\begin{array}{ccc}
0 & 1 & 0 \\
0 & 0 & 1 \\
1 & 0 & 0
\end{array}\right) .
\end{equation}
with the property that $X^{3}= Z^{3}= 1$.
The total Floquet operator in this case is described by
\begin{equation}
    U_{KI_3}= U_{K} U_{I}.
\end{equation}
Generalization to other values of $q$ is straightforward. Yet,  as we show in the next section, a dual-unitary point exists only for $q \leq 4$.

\section{Existence of dual-unitary point for general $q$}\label{sec:dualexist}
In this section, we determine when the generalized kicked Ising chain
introduced in Sec.~II admits a dual-unitary point.
Similar to the
analysis presented in~\cite{Akila_2016}, we get the following partition function 
\begin{align}\label{eq:partxyz}
    Z(N, T) &= \mathrm{Tr}(U_{N}^{T})= 
\sum_{\left\{s_{\tau}= \pm 1\right\}} \prod_{\tau=1}^{t} \braket{s_{\tau+1}|e^{-i H_{k}}e^{-i H_{I}}|s_{\tau}}.
\end{align}
Here $\{\mathbf{s}_{1},\mathbf{s}_{2}, \cdots, \mathbf{s}_{t}\} \equiv \{s_{\tau,j}\}$ which is product of diagonal basis of spin
variable such that $\sigma_{j}^{z}\ket{\mathbf{s}}= s_{j}\ket{\mathbf{s}}$. The action of the kicked part, $U_{K}$, is easily determined. We can simply use the result given in \cite{Akila_2016} with $\varphi= 0$ to obtain the relevant contribution. First, we note that full $U_{K}$ factorizes and it suffices to find the relevant contribution of the subpart, kick at $i-$th site, $U_{K}^{i}= e^{-i b \sigma_{i}^{x}}$. The relevant contribution to the matrix elements are 
evaluated as 
\begin{align}
    \braket{\uparrow|U_{K}^{i}|\uparrow}&= \cos(b)= e^{-iK}e^{\eta} \nonumber \\
    \braket{\downarrow|U_{K}^{i}|\downarrow}&= \cos(b)= e^{-iK}e^{\eta}\nonumber \\
    \braket{\downarrow|U_{K}^{i}|\uparrow}&=\braket{\uparrow|U_{K}^{i}|\downarrow} = -i\sin(b)
\end{align}
where the constant $K$ and $\eta$ are defined as  
\begin{equation}
    e^{-4 i K}=1-\frac{1}{x^2}, \quad e^{4 \eta}=x^2\left(x^2-1\right)
\end{equation}
where $x= \sin b$.  
The total partition is rewritten as 
\begin{widetext}
   \begin{align}\label{z}
 Z(N, T)    = \sum_{\left\{s_{n, \tau}= \pm 1\right\}} \exp \left(-i \sum_{n=1}^N \sum_{\tau=1}^T J s_{n, t} s_{n+1, \tau}+K s_{n, \tau} s_{n, \tau+1}+h s_{n, \tau}+i \eta + J_{x} \right) 
\end{align} 
\end{widetext}

The partition function~\eqref{z} respects symmetry under the exchanging
$n \leftrightarrow \tau $ and also $J \leftrightarrow K$ 
The property of duality between particle and time direction is 
 referred to as the
\emph{dual unitarity}.

We now 
generalize the Ising type model, which has local dimension $q=2$, to models, which has local dimension $q > 2$. 
To achieve this, we study the following Hamiltonian:
\begin{align}\label{eqn:ising3}
    H_{I}& = \sum_{i=1}^{N} J ({Z}_{i}{Z}^{\dagger}_{i+1} +{Z}^{\dagger}_{i}{Z}_{i+1})+ \sum_{i}h_{i}(Z_{i}+ Z^{\dagger}_{i})\nonumber \\
    H_{K}& = \sum_{i=1}^{N} b\, (X_{i}+X^{\dagger}_{i}).
\end{align}
Here, $h_{i}$ can be taken to be generic Gaussian random numbers. Denoting $\omega= e^{2 \pi i/q}$ we can write the explicit representation of $Z$ and $X$ as 
\begin{equation}\label{eqn:zqxq}
Z=\left(\begin{array}{cccc}
1 & 0 & \cdots & 0 \\
0 & \omega & \cdots &0  \\
\vdots & \vdots & \vdots & \vdots  \\
0 & 0 & \cdots&\omega^{q-1}
\end{array}\right), \quad X=\left(\begin{array}{ccccc}
0 & 1 & 0 & \cdots&0 \\
0 & 0 & 1&\cdots &0\\
\vdots & \vdots & \vdots & \vdots &\cdots \\
1 & 0 & 0 &\cdots &0 
\end{array}\right) .
\end{equation}
with $X^{q} = Z^{q}= I$.

The total Floquet operator in this case is
\begin{equation}\label{kick}
    U_{KI3}= U_{HI}U_{K}
\end{equation}
In the following, we evaluate matrix elemets of this kicked model~\eqref{kick}. To this end, 
we 
introduce the Fourier basis as
\begin{equation}
    \ket{f}= \frac{1}{\sqrt{q}} \sum_{n=0}^{q-1} \omega^{n f} \ket{n}.
\end{equation}
One can verify that
\begin{align}
    X\ket{f}= \omega^{k}\ket{f} \\
  X^{\dagger}\ket{f}= \omega^{-k}\ket{f},    
\end{align}
from which we have
\begin{align}
 (X+X^{\dagger})\ket{f}= (\omega^{f}+\omega^{-f}) \ket{f}= 2 \cos\left(\frac{2\pi f}{q}\right)\ket{f} \nonumber \\ 
 \implies \exp[-ib (X+X^{\dagger})]\ket{f}= \exp\left(-2ib\cos\left(\frac{2\pi f}{q}\right)\right)\ket{f}.
\end{align}
Taking
this fact into consideration, we obtain
\begin{align}
    &\braket{m| \exp[-ib (X+X^{\dagger})] | n} \nonumber \\
    &= \sum_{k,k'} \braket{m|f} \exp\left[-2ib\cos\left(\frac{2\pi f}{q}\right)\right]\,\, \braket{f|n} \\
    &= \frac{1}{q} \sum_{f=0}^{q-1}\omega^{f(m-n)}e^{-2 ib\cos\left(\frac{2\pi f}{q}\right)}.
\end{align}
We note that the off diagonal matrix elements are not equal for general $q$. For $q=3$, we explicitly write the matirx:
\begin{widetext}
\begin{equation}
    \exp[-ib (X+X^{\dagger})]= \left(
\begin{array}{ccc}
 \frac{1}{3} e^{-2 i b} \left(1+2 e^{3 i b}\right) & -\frac{1}{3} e^{-2 i b} \left(-1+e^{3 i b}\right) & -\frac{1}{3} e^{-2 i b} \left(-1+e^{3 i b}\right)  \\
 -\frac{1}{3} e^{-2 i b} \left(-1+e^{3 i b}\right) & \frac{1}{3} e^{-2 i b} \left(1+2 e^{3 i b}\right) & -\frac{1}{3} e^{-2 i b} \left(-1+e^{3 i b}\right) \\
 -\frac{1}{3} e^{-2 i b} \left(-1+e^{3 i b}\right) & -\frac{1}{3} e^{-2 i b} \left(-1+e^{3 i b}\right) & \frac{1}{3} e^{-2 i b} \left(1+2 e^{3 i b}\right) \\
\end{array}
\right),
\end{equation}
\end{widetext}
from which we see that the off-diagonal matrix elements are identical.
Likewise, for $q=4$ explicit form of the matrix is given by
\begin{widetext}
    \begin{equation}
    \exp[-ib (X+X^{\dagger})]= \left(
\begin{array}{cccc}
 \frac{1}{4} e^{-2 i b} \left(1+e^{2 i b}\right)^2 & -\frac{1}{4} e^{-2 i b} \left(-1+e^{4 i b}\right) & \frac{1}{4} e^{-2 i b} \left(-1+e^{2 i b}\right)^2 & -\frac{1}{4} e^{-2 i b} \left(-1+e^{4 i b}\right) \\
 -\frac{1}{4} e^{-2 i b} \left(-1+e^{4 i b}\right) & \frac{1}{4} e^{-2 i b} \left(1+e^{2 i b}\right)^2 & -\frac{1}{4} e^{-2 i b} \left(-1+e^{4 i b}\right) & \frac{1}{4} e^{-2 i b} \left(-1+e^{2 i b}\right)^2 \\
 \frac{1}{4} e^{-2 i b} \left(-1+e^{2 i b}\right)^2 & -\frac{1}{4} e^{-2 i b} \left(-1+e^{4 i b}\right) & \frac{1}{4} e^{-2 i b} \left(1+e^{2 i b}\right)^2 & -\frac{1}{4} e^{-2 i b} \left(-1+e^{4 i b}\right) \\
 -\frac{1}{4} e^{-2 i b} \left(-1+e^{4 i b}\right) & \frac{1}{4} e^{-2 i b} \left(-1+e^{2 i b}\right)^2 & -\frac{1}{4} e^{-2 i b} \left(-1+e^{4 i b}\right) & \frac{1}{4} e^{-2 i b} \left(1+e^{2 i b}\right)^2 \\
\end{array}
\right)
\end{equation}
\end{widetext}

In general, we can write following for $j$-th kick

\begin{align} \label{eqn:xqkick}
    \braket{m|e^{-i b (X_{j} + X_{j}^{\dagger})}|m} &=  \frac{1}{q} \sum_{k=0}^{q-1}e^{-2 ib\cos\left(\frac{2\pi k}{q}\right)}\equiv A \nonumber \\
     \braket{m|e^{-i b (X_{j} + X_{j}^{\dagger})}|n} &=\frac{1}{q} \sum_{k=0}^{q-1}\omega^{k(m-n)}e^{-2 ib\cos\left(\frac{2\pi k}{q}\right)} \nonumber\\
     &\equiv B(m-n) \, , m \neq n.
\end{align}
Note that $B(m-n)$ is independent on $m-n$, allowing us to write
$B(m-n) \equiv B$. 
To proceed, we 
define
\begin{equation}\label{eqn:alphabeta}
    \alpha= -i \log(B)\quad \beta= -i (\log(A) - \log(B))
\end{equation}
This allows us to write the contribution as a single exponential
\begin{align}
    \braket{m|e^{-i b (X_{j} + X_{j}^{\dagger})}|n} &=  e^{-i (\alpha + \beta \delta_{mn})}.
\end{align}
we can check if $m = n$ we get $A$ and $B$ otherwise. 
Since
\begin{align}
    \delta_{mn} &= \frac{1}{q} \sum_{k=0}^{q-1}(s_{m}s^{*}_{n})^{k}, \quad (s_i=\omega^i),
\end{align}
we can write 
the contribution of the kick as 
\begin{align}
    \braket{m|e^{-i b (X_{j} + X_{j}^{\dagger})}|n} &=  e^{-i (\alpha + \beta \frac{1}{q}\sum_{k=0}^{q-1}(s_{m}s^{*}_{n})^{k})} \nonumber \\
  \braket{m|e^{-i b (X_{j} + X_{j}^{\dagger})}|n} &=  e^{-i (k_{0} + k_{1}( s_{m}s_{n}^{*} + s^{*}_{m}s_{n}))} 
\end{align}
with $k_{0}= \alpha+ \frac{\beta}{q} \sum_{k=0, k\neq 1,q-1}^{q-2}( s_{m}s_{n}^{*})^{k} $ and $k_1= \beta/q$.
The total contribution due to the kick 
is described by
   \begin{align}
    \braket{s_{\tau}|U_{K}|s_{\tau+1}} &=  \exp\left[\sum_{j=1}^{N}\sum_{\tau=1}^{T}-i (k_{0} + k_{1}(2 \Re(s_{j,\tau}s_{j,\tau+1}) ) \right] 
\end{align} 

For $U_{I}$ part, we obtain 
\begin{align}
    \braket{s_{\tau}|U_{I}|s_{\tau+1}} &= \nonumber \\ &\exp\left[\sum_{j=1}^{N}\sum_{\tau=1}^{T}-i ( J (2 \Re(s_{j,\tau}s_{j+1,\tau}) + h_{j} \Re(s_{i,\tau})) \right] 
\end{align} 

where we explicitly  have $k_{1}$ as
\begin{equation}
    k_{1}= - \frac{i}{q} \log\left(\frac{A}{B}\right)= - \frac{i}{q} \log\left(\frac{ \sum_{k=0}^{q-1}e^{-2 ib\cos\left(\frac{2\pi k}{q}\right)}}{ \sum_{k=0}^{q-1}\omega^{k(m-n)}e^{-2 ib\cos\left(\frac{2\pi k}{q}\right)}}\right) 
\end{equation}
The denominator represents the off diagonal elements, $(U_{K})_{mn}$, of matrix $U_{K}$. For example if we take $q=3$ then $m-n$ can be $1,2$ but it can be easily verified that the value of $(U_{K})_{mn}$ is same in both the cases. For $q \geq 4$, however, 
%
this is not the case and off diagonal elements are not same in general. 

 Now demanding that $k_{1}$ is real will give us the dual-unitary point. To obtain that, we need to find specific value of $b$ such that following holds
 \begin{equation}\label{eq:bsolve}
     \left|\sum_{k=0}^{q-1}e^{-2 ib\cos\left(\frac{2\pi k}{q}\right)}\right| = \left|\sum_{k=0}^{q-1}\omega^{k(m-n)}e^{-2 ib\cos\left(\frac{2\pi k}{q}\right)}  \right| \quad \forall \, m-n \in \{1, \cdots,q-1\}
 \end{equation}

 Solving the above equation would give us $b$ that will correspond to dual-unitary point.

 \begin{table}[H]
     \centering
     \begin{tabular}{|c|c|}
     \hline
         $q$ & $b$ \\\hline
         3 & $\frac{2 \pi }{9}$, $\frac{4 \pi }{9}$ \\\hline
         4 & $\frac{ \pi }{4}$, $\frac{ 3\pi }{4}$ \\\hline
         5& 0.717292 (1), 0.910766 (2), $\cdots$ \\\hline
          6& 0.716109 (1), 0.935929 (2), $\cdots$ \\\hline
         10 & 0.717348 (1),0.920599 (2), $\cdots$ \\\hline
         100 & 0.717348 (1), 0.920592 (2), $\cdots$ \\\hline
         1000 & 0.717348 (1), 0.920592 (2),$\cdots$\\\hline
     \end{tabular}
     \caption{Value of $b$ at which we might have dual unitarity for some values of $q$. Value in the bracket for $q=5,6$ corresponds to the value of $(m-n)$. Till $q=4$ there is a value of $b$ such that Eq. \eqref{eq:bsolve} is true while we do not find that this is the case for larger $q \geq 5$. }
     \label{tab:qdualpoint}
 \end{table}

 For $q \leq 4$ it is possible to obtain the result analytically. However, for large $q$ an asymptotic analysis is possible. First, consider the sum 
 \begin{equation}
     \mathcal{S}= \sum_{k=0}^{q-1}e^{\frac{2 \pi i kr}{q}}e^{-2 ib\cos\left(\frac{2\pi k}{q}\right)} , r \in\{0,1,\cdots\}.
 \end{equation}
For large $q$, we replace the summation by an integral
\begin{equation}
    \mathcal{S} \xrightarrow[]{q\rightarrow\infty} \frac{q}{2\pi} \int^{2\pi}_{0} e^{ir\theta} e^{-2 i b \cos(\theta)} \mathrm{d\theta}= q i^{-r} J_{r}(2b)
\end{equation}
where $J_{n}(x)$ is the Bessel's function. Thus, for large $q$ we can write Eq. \eqref{eq:bsolve} as 
\begin{equation}\label{eq:bvalue}
    |J_{0}(2b)|= |J_{m-n}(2b)|
\end{equation}
It is then easy to see that it reproduces the values for large $q$ shown in the Table \eqref{tab:qdualpoint}. However, the table suggests that there is \emph{no unique value} of $b$ for $q\geq4$ such that Eq. \eqref{eq:bsolve} holds, for all $m-n \in \{1, \cdots,q-1\}$ . We can show it by assuming that there exist $c$ such that $|J_{0}(2b)|= |J_{m-n}(2b)|= c$ holds. Then we note the following property of Bessel's function 
\begin{equation}
    J_0^2(z)+2 \sum_{k=1}^{\infty} J_k^2(z)=1
\end{equation}
Substituting $|J_{0}(2b)|= |J_{m-n}(2b)|= c$ we have
\begin{equation}
    c^{2} + 2 \sum_{r=1}^{\infty}c^{2}=1
\end{equation}
However, we have 
\begin{equation}
    \sum_{r=1}^{\infty}c^{2}= \begin{cases}
        \infty \quad c \neq 0\\
        0 \quad c = 0
    \end{cases}
\end{equation}
which gives a contradiction. Thus there is no single value of $ b$ (substituting $z=2b$) such that Eq. \eqref{eq:bvalue} holds. 
This demonstrates the absence of dual-unitary points for
 for $q>4$. 
A way is to understand it is as follows.

The dual-unitary point is such that, $J= \frac{2 \pi}{3},\frac{4 \pi}{3}$, and the modulus of matrix elements of $U_{K}$, i.e. $(U_{K})_{mn}$ is the same which also translates to the fact that a unique value of $b$ exists such that Eq. \eqref{eq:bsolve} holds. 
However, using the argument above and Table \eqref{tab:qdualpoint}, we see that for $q \geq 5$ there is no unique value of $b$ such that all $(U_{K})_{mn}$ have same modulus. So the dual-unitary point does not exist for $q \geq 5$. Explicit treatment for the case of $q=3$ can be found in Appendix~\ref{append:ptduality}.

In summary, a physical way to understand the absence of dual-unitary points for $q\geq 5$ is the following.
Dual unitarity requires the Floquet gate to behave, in a suitable basis, like a perfectly scrambling object whose matrix elements all have identical modulus. In the present model, this translates into the requirement that all transition amplitudes generated by the kick operator $U_K$ have equal magnitude.
For $q=2,3,4$, the discrete Fourier structure of the generalized Pauli operators is sufficiently constrained that a single value of the kick strength b can equalize all matrix elements simultaneously. In other words, destructive and constructive interference between different clock states can still be tuned uniformly.
For larger local Hilbert space dimensions, however, the number of independent off-diagonal amplitudes increases. The different Fourier components then oscillate with inequivalent phases, and a single parameter $b$ is no longer sufficient to enforce equal modulus for all matrix elements simultaneously. From this perspective, the breakdown of dual unitarity for $q\geq 5$ originates from an over constrained interference condition in the enlarged local Hilbert space.
The asymptotic Bessel-function analysis makes this obstruction explicit: different hopping sectors correspond to different Bessel functions $J_r(2b)$, whose magnitudes cannot all coincide at a single value of $b$.
\section{Solvable states and the setup}\label{sec4}
We consider a general 3-level state
\begin{widetext}
   \begin{align} \label{eqn:gstate}
    \ket{\psi_{\theta, \chi, \phi_{1},\phi_{2}}}&= \bigotimes_{i=1}^{L}\Bigg(\cos \left(\frac{\theta_{i}}{2}\right)\ket{0}+e^{i \alpha_{i}}\sin \left(\frac{\theta_{i}}{2}\right) \cos \left(\frac{\chi_{i}}{2}\right) \ket{1} +e^{i \beta_{i}}\sin \left(\frac{\theta_{i}}{2}\right) \sin \left(\frac{\chi_{i}}{2}\right)\ket{2}\Bigg).
\end{align} 
\end{widetext}

We then define three classes of states. The first class corresponds to parameters $\theta_{i}= \Theta= 2\arccos \frac{1}{\sqrt{3}}$, $\chi_{i}= \frac{\pi}{2}$ $\forall i \in [1,L]$ and is given as follows
\begin{equation}
    \mathcal{T}_{1}= \ket{\psi_{\Theta, \pi/2, \alpha,\beta}}.
\end{equation}

 The second class corresponds parameters $\theta_{i}= 0$ $\forall i\in [1,L]$ and is given as follows 
 \begin{equation}
    \mathcal{T}_{2}= \ket{\psi_{0, \chi, \alpha,\beta}}\,\forall \chi_{i} \in [0,2\pi].
\end{equation} 

Finally the third class corresponds to the parameters $\theta_{i}= \pi, \chi \in \{0,\pi\}$ and is given as follows
 \begin{align}
    \mathcal{T}_{3}&= \ket{\psi_{\pi, \chi, \alpha,\beta}}.
\end{align} 
We can immediately see two special case of $\mathcal{T}_{3}$ with $\chi=0,\pi$ are $\ket{1}$ and $\ket{2}$ respectively. 

After $t$-kicks of the Floquet operator, the evolved state is given by 
\begin{equation}
    \ket{\psi_{\theta, \chi, \alpha,\beta}(t)}= (U_{KI3}[\textbf{h}])^{t} \ket{\psi_{\theta, \chi, \alpha,\beta}}.
\end{equation}
At this point, we note that the time evolution of $\mathcal{T}_{2}$ and $\mathcal{T}_{3}$ state are related to that of $\mathcal{T}_{1}$ state with $\chi=0,\pi.$

After one kick can have following relation for the $\mathcal{T}_{2}$ state
\begin{align}
    U_{KI_3} \ket{\psi_{0, \chi, \alpha,\beta}}& \simeq \ket{\psi_{\Theta, \frac{\pi}{2}, \frac{-2\pi}{3},\frac{-2\pi}{3}}}, 
\end{align}
where $\simeq$ implies that they are equal upto a global phase that we can safely ignore. 
Similarly for the two states belonging to $\mathcal{T}_{3}$ state we have 
\begin{align}
    U_{KI_3} \ket{\psi_{\pi,0, \alpha,\beta}} &\simeq \ket{\psi_{\Theta, \pi/2, 2\pi /3,0}},\\
      U_{KI_3} \ket{\psi_{\pi,\pi, \alpha,\beta}} &\simeq \ket{\psi_{\Theta, \pi/2, 0,2\pi /3}}.
\end{align}
These relations then imply that we have following 
    \begin{align}
    & \ket{\psi_{0, \chi, \alpha,\beta}(t)} \simeq \ket{\psi_{\Theta, \frac{\pi}{2}, -\frac{2\pi}{3},-\frac{2\pi}{3}}(t-1)},  \\ 
      & \ket{\psi_{\pi,0, \alpha,\beta}(t)} \simeq \ket{\psi_{\Theta, \frac{\pi}{2}, \frac{2\pi}{3},0}(t-1)},\\
     & \ket{\psi_{\pi,\pi, \alpha,\beta}(t)}\simeq \ket{\psi_{\Theta, \frac{\pi}{2}, 0,\frac{2\pi}{3}}(t-1)}.
\end{align}
The above relations imply that the $t-$th time evolved $\mathcal{T}_{2}$ and $\mathcal{T}_{3}$ type state is related to some $(t-1)$-th time evolved $\mathcal{T}_{1}$ state. 

The entanglement entropy is obtained using reduced density matrix. 
This matrix is given by
\begin{equation}\label{eqn:rhogen}
    \rho_{A}(t)= \mathrm{tr}_{\mathcal{H}_{L-N}}[\ket{\psi_{\theta, \chi, \phi_{1},\phi_{2}}(t)} \bra{\psi_{\theta, \chi, \phi_{1},\phi_{2}}(t)} ],
\end{equation}
where the partial trace is obtained by tracing out the $L-N$ qubits,  corresponding to Hilbert space $\mathcal{H}_{L-N}$. The entanglement content of $\rho_{A}(t)$ is quantified by the R\'enyi entropies $S^{(\alpha)}_{A}(t)$ which is given by 
\begin{equation}
    S^{(\alpha)}_{A}(t) = \frac{1}{1-\alpha} \log \tr [(\rho_{A}(t))^{\alpha}].
\end{equation}
Of interest is also the von Neumann entropy which is given by 
\begin{equation}
    S_{A}(t)= \lim_{\alpha \rightarrow 1} S^{(\alpha)}_{A}(t)= - \tr (\rho \log(\rho)).
\end{equation}

The quantity of interest for us is $\tr[(\rho_{A}(t))^{n}]$, which we can now conveniently write using Eq. \eqref{eqn:gstate}  and \eqref{eqn:rhogen} as follows 
\begin{widetext}
    \begin{align}
\tr[(\rho_{A}(t))^{n}] &=\sum_{\{\mathbf{a}_{i}, \mathbf{b}_{i}\}} \langle  \psi_{\theta, \chi, \phi_{1},\phi_{2}}| (U_{KI_{3}}[\mathbf{h}])^{-t} |\mathbf{a}_{1},\mathbf{b}_{2} \rangle \langle \mathbf{a}_{1},\mathbf{b}_{1}| (U_{KI_{3}}[\mathbf{h}])^{t})|\psi_{\theta, \chi, \phi_{1},\phi_{2}} \rangle \nonumber \\
&\times \langle  \psi_{\theta, \chi, \phi_{1},\phi_{2}}| (U_{KI_{3}}[\mathbf{h}])^{-t} |\mathbf{a}_{2},\mathbf{b}_{3} \rangle \langle \mathbf{a}_{2},\mathbf{b}_{2}| (U_{KI_{3}}[\mathbf{h}])^{t})|\psi_{\theta, \chi, \phi_{1},\phi_{2}} \rangle \nonumber\\
&~~~~~~~~~~~~~~~~~~~~~~~~~~~~~~~~~~~~~~\vdots \nonumber \\
&\times \langle  \psi_{\theta, \chi, \phi_{1},\phi_{2}}| (U_{KI_{3}}[\mathbf{h}])^{-t} |\mathbf{a}_{n},\mathbf{b}_{1} \rangle \langle \mathbf{a}_{n},\mathbf{b}_{n}| (U_{KI_{3}}[\mathbf{h}])^{t})|\psi_{\theta, \chi, \phi_{1},\phi_{2}} \rangle
\end{align}
\end{widetext}

We find the contribution above we basically need to find the following building block
\begin{equation}
    \langle \mathbf{a},\mathbf{b}| (U_{KI_{3}}[\mathbf{h}])^{t})|\psi_{\theta, \chi, \phi_{1},\phi_{2}} \rangle
\end{equation}

It is easy to evaluate it by inserting $t$ resolutions of identity 
\begin{widetext}
    \begin{align}
   \langle \mathbf{a},\mathbf{b}| (U_{KI_{3}}[\mathbf{h}])^{t})|\psi_{\theta, \chi, \phi_{1},\phi_{2}} \rangle & = \sum_{\{s_{\tau}\}} \prod_{\tau=1}^{t-1} \langle \mathbf{a},\mathbf{b}|U_{KI_{3}}[\mathbf{h}]|\textbf{s}_{\tau} \rangle \nonumber \times \langle \mathbf{a},\mathbf{b}| U_{KI_{3}} [\mathbf{h}] | \textbf{s}_{\tau}\rangle \langle \textbf{s}_{1} | \psi_{\theta, \chi, \phi_{1},\phi_{2}} \rangle.
\end{align}
\end{widetext}

Then we obtain the matrix elements
\begin{widetext}
    \begin{align}
    \langle \mathbf{s} | \psi_{\theta, \chi, \phi_{1},\phi_{2}} \rangle = \prod_{j=1}^{L} (\cos(\theta_{j}/2) \delta_{s_{j},0} + e^{i \alpha_{j}} \sin(\theta_{j}/2)\cos(\chi_{j}/2) \delta_{s_{j},1} +e^{i \beta_{j}} \sin(\theta_{j}/2)\sin(\chi_{j}/2) \delta_{s_{j},2} ), 
\end{align}
\end{widetext}
where we denote the basis states as $\ket{0},\ket{1}$ and $\ket{2}$ with eigenvalues $1, \omega, \omega^{2}$ respectively,
and 

\begin{align}
\braket{\mathbf{s} | U_{KI_{3}} | \mathbf{r} } &= \exp\left[\sum_{j=1}^{L}-i ( J (2 \Re(s_{j}s_{j+1}) + h_{j} 2 \Re(s_{i})) \right]  \nonumber \\
&\times \exp\left[\sum_{j=1}^{L}-i (k_{0} + k_{1}(2 \Re(s_{j}r_{j}) ) \right] , 
\end{align}

where $k_{0}=\alpha+ \beta/3 $ and $k_1= \beta/3$ and $\alpha, \beta$ are given by Eq. \eqref{eqn:alphabeta}. Notice that $k_{1}= J = \frac{2 \pi}{9},\frac{4 \pi}{9}$ at the dual-unitary point
\begin{equation}
    \braket{s | e^{i b (X + X^{\dagger})} | r}= \frac{1}{3} e^{-2 i b} (1+ 2 e^{3 i b} - 3 \delta_{sr}e^{3 i b}), \, s \in \{1, \omega, \omega^{2} \}. 
\end{equation}
Putting this all together, we  have 
\begin{widetext}
\begin{align}
  &\langle \mathbf{a},\mathbf{b}| (U_{KI_{3}}[\mathbf{h}])^{t})|\psi_{\theta, \chi, \phi_{1},\phi_{2}} \rangle=  \sum_{\{s_{j,\tau}\}} \exp\left[\sum_{j=1}^{L}\sum_{\tau=1}^{t}-i ( J (2 \Re(s_{j,\tau}s_{j+1,\tau}) + 2 h_{j} \Re(s_{i,\tau})) - i\sum_{j=1}^{L}\sum_{\tau=1}^{t-1} (k_{0}+J(2 \Re(s_{j,\tau}s_{j,\tau+1}) ) \right] \nonumber \\ 
&  \times 
 \exp\left[\sum_{j=1}^{N}-i (k_{0} + J (2 \Re(s_{j,t}a_{j}) ) \right] \times \exp\left[\sum_{j=N+1}^{L}-i (k_{0} + J(2 \Re(s_{j,t}b_{j-N}) ) \right] \nonumber \\
&\times \prod_{j=1}^{L} (\cos(\theta_{j}/2) \delta_{s_{1,j},0} + e^{i \alpha_{j}} \sin(\theta_{j}/2)\cos(\chi_{j}/2) \delta_{s_{1,j},1} +e^{i \beta_{j}} \sin(\theta_{j}/2)\sin(\chi_{j}/2) \delta_{s_{1,j},2} ). 
\end{align}
\end{widetext}
Note that $\exp[- i k_{0} L t]$ which can evaluate the self dual point, using Eq. \eqref{eqn:alphabeta}, and obtain $\exp[- i k_{0} L t]= \sqrt{3}^{L t} \times (\text{phase})$. We ignore this global phase factor in our calculations.

Next, we use the replica trick and see that there are basically $2n$ different building block. It is easily done by noting using $\nu$ for parameterization which basically goes from $\nu = 1, \cdots,  2n$ such that $\nu \leq n$ corresponds to forward time evolution and $\nu \geq n+1$ corresponds to backward time evolution. Finally, we obtain, after summing over $\{\mathbf{a}_{j},\mathbf{b}_{j} \}$

\begin{widetext}
  \begin{align}
&\tr[(\rho_{A}(t))^{n}] \nonumber \\
&= \frac{1}{3^{nLt}}\sum_{\{s_{\nu,\tau,j}\}} \exp\left[\sum_{j=1}^{L} \sum_{\nu=1}^{2n} \mathrm{sgn}(n-\nu)\left(  \sum_{\tau=1}^{t}-i ( J (2 \Re(s_{j,\tau}s_{j+1,\tau}) + 2 h_{j} \Re(s_{i,\tau})) - i \sum_{\tau=1}^{t-1} (J(2 \Re(s_{j,\tau}s_{j,\tau+1}) ) \right )\right ] \nonumber \\
& \times \prod_{\nu=1}^{n} \left\{ \prod_{j=1}^{N}\left( 1+ s_{\nu,t,j} s^{-1}_{\nu+n,t,j} + (s_{\nu,t,j} s^{-1}_{\nu+n,t,j})^{2} \right ) \prod_{j= N+1}^{L}\left( 1+ s_{\nu,t,j} s^{-1}_{n+1 + \textrm{mod}(\nu-2,n),t,j} + (s_{\nu,t,j} s^{-1}_{n+1 +\textrm{mod}(\nu-2,n),t,j})^{2} \right ) \right\} \nonumber \\
&\times \prod_{\nu=1}^{2n} \prod_{j=1}^{L} (\cos(\theta_{j}/2) \delta_{s_{\nu,1,j},0} + e^{i \alpha_{j} \mathrm{sgn}(n-\nu)} \sin(\theta_{j}/2)\cos(\chi_{j}/2) \delta_{s_{\nu,1,j},1} +e^{i \beta_{j}\mathrm{sgn}(n-\nu)} \sin(\theta_{j}/2)\sin(\chi_{j}/2) \delta_{s_{\nu,1,j},2} ), 
\end{align}   
\end{widetext}
where $\textrm{sgn}(x)$ is the sign function with the convention $\textrm{sgn}(0)=1$  and $\textrm{mod}(m,n)= m\, \textrm{mod}\, n$.
Where we use the identity

\begin{align}
   \sum_{a\in \{1,\omega, \omega^{2}\}} &\exp \left [ -i \frac{4 \pi}{9 } \Re (a (s-r)) \right]= 3 \delta_{sr}, \nonumber \\ 
    &= 1+sr^{-1} + (sr^{-1} )^{2} , \quad s, r \in \{1,\omega, \omega^{2}\}.
\end{align}

The above expression can be easily understood by considering the tensor product space $\mathcal{H}_{t}^{\otimes 2n}$ through their matrix elements in the computation basis
\begin{widetext}
\begin{align}\label{eqn:tmatelement}
  & \braket{ \{\mathbf{s}_{\nu,\tau}\} | \mathbb{T} | \{\mathbf{r}_{\nu,\tau}\}} = \frac{1}{3^{(t-1)n}}\exp\left[ \sum_{\nu=1}^{2n} \mathrm{sgn}(n-\nu)\left(  \sum_{\tau=1}^{t}-i ( J (2 \Re(s_{\nu,\tau}s_{\nu+1,\tau}) + h_{j} \Re(s_{\nu,\tau})) - i \sum_{\tau=1}^{t-1} (J(2 \Re(s_{\nu,\tau}s_{\nu,\tau+1}) ) \right )\right ] \nonumber \\
& \times \prod_{\nu=1}^{n} \left\{ \frac{\left( 1+ s_{\nu,t} s^{-1}_{\nu+n,t} + (s_{\nu,t} s^{-1}_{\nu+n,t})^{2} \right )}{3}  \right\} \nonumber \\
&\times \prod_{\nu=1}^{n}  (\cos(\theta/2) \delta_{s_{\nu,1},1} + e^{i \alpha\, \mathrm{sgn}(n-\nu)} \sin(\theta/2)\cos(\chi/2) \delta_{s_{\nu,1},\omega} +e^{i \beta\, \mathrm{sgn}(n-\nu)} \sin(\theta/2)\sin(\chi/2) \delta_{s_{\nu,1},\omega^2} ) 
\end{align}

\text{and}
\begin{align}\label{eqn:rmatelement}
& \braket{ \{\mathbf{s}_{\nu,\tau}\} | \mathbb{R} | \{\mathbf{r}_{\nu,\tau}\}} = \frac{1}{3^{(t-1)n}} \exp\left[ \sum_{\nu=1}^{2n} \mathrm{sgn}(n-\nu)\left(  \sum_{\tau=1}^{t}-i ( J (2 \Re(s_{\nu,\tau}s_{\nu+1,\tau}) + 2 h_{j} \Re(s_{\nu,\tau})) - i \sum_{\tau=1}^{t-1} (J(2 \Re(s_{\nu,\tau}s_{\nu,\tau+1}) ) \right )\right ] \nonumber \\
& \times \prod_{\nu=1}^{n} \left\{ \frac{\left( 1+ s_{\nu,t} s^{-1}_{n+1 + \textrm{mod}(\nu-2,n),t} + (s_{\nu,t} s^{-1}_{n+1 +\textrm{mod}(\nu-2,n),t})^{2} \right )}{3} \right\} \nonumber \\
&\times \prod_{\nu=1}^{2n}  (\cos(\theta/2) \delta_{s_{\nu,1},1} + e^{i \alpha \mathrm{sgn}(n-\nu)} \sin(\theta/2)\cos(\chi/2) \delta_{s_{\nu,1},\omega} +e^{i \beta \mathrm{sgn}(n-\nu)} \sin(\theta/2)\sin(\chi/2) \delta_{s_{\nu,1},\omega^2} ), 
\end{align} 
\end{widetext}
where we use $\mathbb{R}[h] \equiv \mathbb{R}_{\theta_{i},\chi_{i},\alpha_{i}, \beta_{i}}[h]$ and $\mathbb{T}[h] \equiv \mathbb{T}_{\theta_{i},\chi_{i},\alpha_{i}, \beta_{i}}[h]$ for brevity.

\begin{align}\label{eqn:trrhofull}
\tr[(\rho_{A}(t))^{n}]= \tr\left[\left( \prod_{i=1}^{N} \mathbb{T}[h_{j}] \right) \left( \prod_{i=N+1}^{L} \mathbb{R}[h_{j}]  \right) \right]. 
\end{align}

Now if we reinterpret $\mathcal{H}_{t}^{\otimes 2n}$ instead as tensor product of two copies of $\mathcal{H}_{nt}$. The first $n$ copies are group together and last $n$ copies are grouped together and one can write the basis elements of $\mathcal{H}_{t}^{\otimes 2n}$ as 
\begin{equation}
    \ket{\{r_{j,\tau}\}_{1 \leq \tau \leq t}^{1 \leq j \leq 2n}}=  \ket{\{r_{j,\tau}\}_{1 \leq \tau \leq t}^{1 \leq j \leq n}} \otimes  \ket{\{r_{j,\tau}\}_{1 \leq \tau \leq t}^{n+1 \leq j \leq 2n}}.
\end{equation}

Then we can see that $\mathbb{T}$ and $\mathbb{R}$ only differs via a cyclic permutation of copies of $\mathcal{H}_{t}$ composing of negative-time space. Thus we have
\begin{equation}   
\mathbb{R}[h]= \mathbb{P}\mathbb{T}[h], \mathbb{P}^{\dagger}
\end{equation}
where we define 
\begin{equation}
    \mathbb{P}= \mathbb{1} \otimes \prod_{\nu=1}^{n}\prod_{\tau=1}^{t} P_{(\nu,\tau),(\nu-1,\tau)}.
\end{equation}
where $P_{ij} = \frac{1}{3}\sum_{a,b} X_{i}^{a}Z_{j}^{b} \otimes Z_{j}^{-b} X_{i}^{-a}$
is a transposition operator for qutrits. So $\mathbb{P}$ only acts on the negative-time space. 

Writing Eq. \eqref{eqn:tmatelement} in matrix form, we can see that the transfer matrix is a tensor product of single copy transfer matrices\footnote{This time the labeling is again from 1  to $\nu$ for each copy and the product is understood as acting non-trivially only on the $\nu$-th copy. }

\begin{equation}
    \mathbb{T}[h] = \prod_{\nu=1}^{n} \mathbb{T}^{(\nu)}[h].
\end{equation}

The single copy matrix can thus be written as 
\begin{equation}
    \mathbb{T}^{(\nu)}[h]= \mathbb{H}_{\nu,1}[\theta,\chi] .\mathcal{P}^{(Z)}_{\nu,t}. \mathbb{U}_{\alpha,\beta}^{(\nu)}[h]. 
\end{equation}
Here $\mathbb{H}[\theta,\chi]$ is Hermitian matrix, $\mathcal{P}$ is a projector and $\mathbb{U}[h] $ is unitary (only because we are at the dual-unitary point) are defined as follows (after conveniently dropping the global phase factor)-

\begin{widetext}
    \begin{align}
    \mathbb{H}_{\nu,1}[\theta,\chi]&= 3 \left( (\cos(\frac{\theta}{2}) P^{(0)}_{\nu, 1} +  \sin(\frac{\theta}{2}) \cos(\frac{\chi}{2}) P^{(1)}_{\nu, 1} + \sin(\frac{\theta}{2})\sin(\frac{\chi}{2}) P^{(2)}_{\nu, 1} \right )^{\otimes 2}, \\
    \mathcal{P}^{(Z)}_{\nu,\tau}&=  \frac{\mathbb{1} + Z_{\nu,\tau} \otimes Z_{\nu,\tau}^{-1} + Z_{\nu,\tau}^{2} \otimes Z_{\nu,\tau}^{-2}}{3}, \\
    \mathbb{U}_{\alpha,\beta}^{(\nu)}[h]&= U_{\nu,\alpha,\beta}e^{-i h M_{\nu}^{z}}e^{-i J M_{\nu}^{x}} \otimes U^{\dagger}_{\nu,\alpha,\beta}e^{i h M_{\nu}^{z}}e^{i J M_{\nu}^{x}}, \\
    U_{\nu, \alpha,\beta} &= \exp \left[- i J \sum_{\tau=1}^{t-1}( Z_{\nu,\tau} Z^{\dagger}_{\nu,\tau+1} +Z^{\dagger}_{\nu,\tau} Z_{\nu,\tau+1} )  - i \alpha P^{(1)}_{\nu,1} - i \beta P^{(2)}_{\nu,1}  \right],
\end{align}
\end{widetext}
\begin{align}
    M_{\nu}^{Z} &= \sum_{\tau=1}^{t}(Z_{\nu,\tau} + Z^{\dagger}_{\nu,\tau}), \\
    M_{\nu}^{X} &= \sum_{\tau=1}^{t}(X_{\nu,\tau} + X^{\dagger}_{\nu,\tau}), \\
P^{(k)}_{\nu, \tau}&= \frac{1}{3} ( I + \omega^{-k} Z_{\nu,\tau} + \omega^{-2k} Z_{\nu,\tau}^{2}),
\end{align}
which is the projector\footnote{Consider $Z \ket{l}= \omega^{l} \ket{l}$ then we have $P^{k}\ket{l}= \frac{1}{3} (1+ \omega^{k-l}+ \omega^{2(k-l)} )$ which is 1 when $k=l$ so it only projects onto $\omega^{k}$ eigenvalues and is zero otherwise.}. 

\section{Spread of entanglement}\label{sec5}

We now use the properties of transfer matrix $\mathbb{T}$ as derived in appendix \eqref{sec:tprop} and \eqref{sec:tevec}. 
\begin{prop}
Transfer matrix $\mathbb{T}[h]$ has following properties
\begin{itemize}
    \item[\text{(i)}] $|\lambda_{i}| \leq 3^n \max \left(\cos ^{2 n}\left(\frac{\theta}{2}\right),\, \sin ^{2 n}\left(\frac{\theta}{2}\right) \left(\frac{1}{2}(| \cos (\chi)|+1)\right)^n\right) $, $\forall \lambda_{i} \in \textrm{Spec}(\mathbb{T}[h])$.
     \item[\text{(ii)}] If $\lambda$ is an eigenvalues of $\mathbb{T}[h]$ such that $|\lambda|= \lambda_{max}$, then 
     \begin{itemize}
         \item[(a)] $\lambda$ has trivial Jordon block i.e. its geometric and algebraic multiplicities coincide. 
         \item[(b)] The associated left eigenvector $\bra{V_{L}}$ has the following properties
          \begin{align}
              \bra{V_{L}} \mathbb{H}_{\nu,1}[\theta,\chi] &= \lambda_{\mathrm{max}} \bra{V_{L}} \label{eqn:hprop},\\
              \bra{V_{L}} \mathcal{P}^{(Z)}_{\nu,t}[\theta,\chi] &= \bra{V_{L}} \label{eqn:pprop}, \\
              \bra{V_{L}} \mathbb{U}^{(\nu)}_{\alpha,\beta}[h] &= e^{i \phi_{\nu}} \bra{V_{L}} \label{eqn:uprop}.
          \end{align}
     \end{itemize}
\end{itemize}
\end{prop}

We further have 
\begin{prop}
The spectrum of $\mathbb{T}[h]$ is fully determined as follows
\begin{enumerate}
    \item $\text{Spec}(\mathbb{T}[h])= \{0,1\}$.
    \item The geometric multiplicity of the eigenvalue 1 is one. 
\end{enumerate}
\end{prop}
This means that maximum eigenvalue property can be saturated only when $\lambda_{\textrm{max}}=1$ which is obtained for 
\begin{align}
\cos \left(\frac{\theta}{2}\right)&= \sqrt{\frac{1}{3}}, \quad  \sin\left(\frac{\theta}{2}\right) \cos\left(\frac{\chi}{2}\right)= \sqrt{\frac{1}{3}}, \nonumber \\  &\sin\left(\frac{\theta}{2}\right) \sin\left(\frac{\chi}{2}\right)= \sqrt{\frac{1}{3}}
\end{align}
Using which we obtain one such state as
\begin{equation} \label{eqn:solvableval}
 \cos \left(\frac{\theta}{2}\right)= \sqrt{\frac{1}{3}} , \quad  \sin \left(\frac{\theta}{2}\right)= \sqrt{\frac{2}{3}}, \quad \chi = \frac{\pi}{2}.
\end{equation}

Putting everything together we have 
\begin{equation}\label{eqn:san_final}
 S_{A}^{(n)}(t)= \frac{1}{1-n} \log \tr \left[ \left( \prod_{j=1}^{N} \mathbb{T}[h_{j}] \right) \mathbb{P} \left( \prod_{j=N+1}^{L} \mathbb{T}[h_{j}]  \right) \mathbb{P}^{\dagger}  \right]
\end{equation}
Taking the large $L$ limit we obtain
\begin{equation}
\lim_{L \rightarrow \infty} S_{A}^{(n)}(t)= \frac{1}{1-n} \log  \left[ \bra{\mathbb{1}} \mathbb{P}^{\dagger}\left( \prod_{j=1}^{N} \mathbb{T}[h_{j}] \right) \mathbb{P} \ket{\mathbb{1}}   \right].
\end{equation}
If we further take $N \rightarrow \infty$ we obtain

\begin{equation}
\lim_{N \rightarrow \infty}\lim_{L \rightarrow \infty} S_{A}^{(n)}(t)= \frac{2}{1-n} \log  | \braket{\mathbb{1}|\mathbb{P}^{\dagger} |\mathbb{1}}= 2 t\log(3).
\end{equation}

However, it is also possible to evaluate Eq.\eqref{eqn:san_final} for finite $N$. Using calculation as shown in appendix \ref{append:trfiniten} we obtain
\begin{widetext}
    \begin{align}\label{eqn:tfinal}
\langle\Psi|\left(\prod_{j=1}^N \mathbb{T}\left[h_j\right]\right)|\Psi\rangle = \braket{ \Psi|\prod_{\nu=1}^{n}\left[ \prod_{j=1}^{\lfloor N/2 \rfloor-1} [\mathcal{P}_{\nu,t-j}^{(Z)} \mathcal{C}_{\nu,t-j}^{(X)}  \right] [ \mathcal{P}_{\nu,t- \lfloor N/2 \rfloor}^{(Z)}]^{\mathrm{mod}(N,2)}  |\Psi},
\end{align}
\end{widetext}
where $\mathcal{C}_{\nu,\tau}^{(X)}= \frac{1}{3}( \mathbb{1}+ X_{\nu,\tau}\otimes X_{\nu,\tau}+ X_{\nu,\tau}^{\dagger}\otimes X_{\nu,\tau}^{\dagger}).$

We next note the following properties
\begin{widetext}
   \begin{align}
    \bra{p'} \otimes \bra{q'}|\mathbb{1}| \ket{p'} \otimes \ket{q'}= \delta_{pp'}\delta_{qq'}, \\
    \bra{p'} \otimes \bra{q'}|\frac{1}{3}(\mathbb{1}+ Z \otimes Z^{\dagger} + Z^{\dagger} \otimes Z)| \ket{p} \otimes \ket{q}= \delta_{pq} \delta_{p'p}\delta_{q'q}, \\
    \bra{p'} \otimes \bra{q'}|\frac{1}{3}(\mathbb{1}+ Z \otimes Z^{\dagger} + Z^{\dagger} \otimes Z) \frac{1}{3}(\mathbb{1}+ X \otimes X + X^{\dagger} \otimes X^{\dagger})| \ket{p} \otimes \ket{q}= \frac{1}{3} \delta_{p'q'}\delta_{pq}.
\end{align} 
\end{widetext}
which allow us to evaluate Eqn. \eqref{eqn:tfinal} exactly and we obtain 
\begin{equation}
    \langle\Psi|\left(\prod_{j=1}^N \mathbb{T}\left[h_j\right]\right)|\Psi\rangle = \begin{cases}
3^{N(1-n)} &t> \lfloor N/2 \rfloor,\\
3^{2t(1-n)} &t< \lfloor N/2 \rfloor.
\end{cases}
\end{equation}

This also implies that the spectrum of the reduced density matrix is given as
\begin{equation}
   \textrm{Spec}[\rho_{A}(t)]= \{3^{-\min(2t,N)},0\},   
\end{equation}
where $3^{-\min(2t,N)}$ has multiplicity $3^{\min(2t,N\}}$, while $0$ has multiplicity $3^N-3^{\min(2t,N\}}$. This then implies that all the R\'enyi entropies are same.

\section{Numerical study} \label{sec6}
\subsection{Solvable states}
For the numerical results we take $L=18$ qutrits and evaluate von Neuman entanglement entropy of a subsystem of $N=6$ qutrits. As shown in Fig. \eqref{fig:eeim3} we obtain the following for the $Z_{3}$ Ising model 
\begin{equation}
    S(t) = \min(2 t, N) \log(3).
\end{equation}

\begin{figure}[H]
    \centering
    \includegraphics[width=0.5\textwidth]{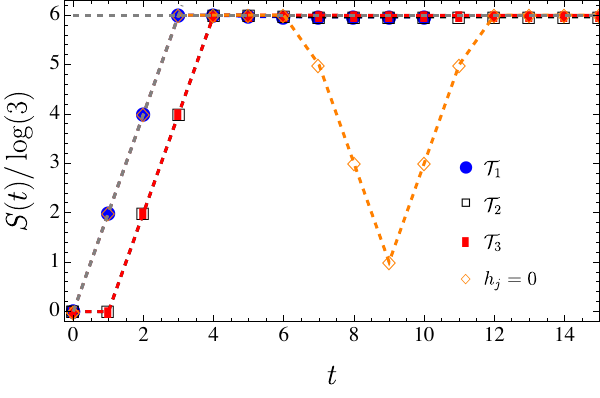}
    \caption{Spread of entanglement for $q= 3$ KPM with $L=18$, Eq. \eqref{eqn:ising3}. We consider the initial states as the three type of solvable states $\mathcal{T}_{1}, \mathcal{T}_{2}$ and $\mathcal{T}_{3}$. Also shown is the results for $\mathcal{T}_{1}$ state but with $h_{j}=0$. The results for $\mathcal{T}_{2}$ and $\mathcal{T}_{3}$ states are delayed by one period.  }
    \label{fig:eeim3}
\end{figure}

Similarly, for the case of for the $Z_{4}$ Ising model we consider $L = 15$ 
qudits (with local Hilbert space dimension of $4$)  and evaluate von Neumann entanglement entropy of a subsystem of $N=6$ qudits. As shown in Fig. \eqref{fig:eeim4}, we obtain results consistent with
\begin{equation}
    S(t) = \min(2 t, N) \log(4).
\end{equation}
As observed previously from Fig. \eqref{fig:eeim3}, the integrable model ($h_{j}=0$) exhibit oscillations after the initial linear growth. It is then natural to ask the fate of these oscillations for weak integrability breaking. We study the effect of weak integrability breaking in Fig. \ref{fig:eeimg3wib} where we study the $q=3$ kicked Potts model starting from a $\mathcal{T}_{1}$ type state and by consider various values of longitudinal field. It is observed that the oscillations reduces as the longitudinal field become stronger.

\begin{figure}[H]
    \centering
    \includegraphics[width= 0.95\linewidth]{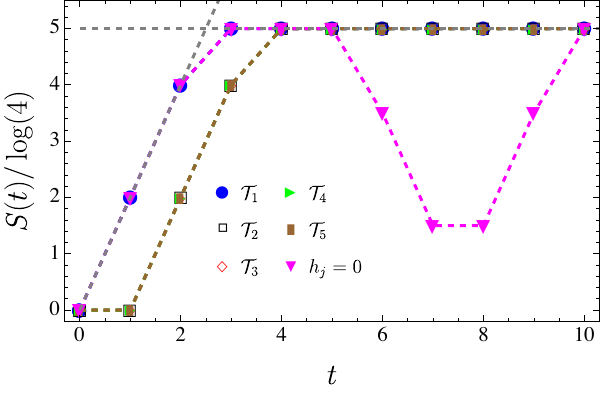}
    \caption{Time evolution of von  Neumann entanglement entropy for $q= 4$ model for a subsystem of $N=5$ qudits in a system of $L=15$ qudits KPM. The initial states are three type of solvable state $\mathcal{T}_{1}, \mathcal{T}_{2}$ and $\mathcal{T}_{3}$. Also shown is the results for $\mathcal{T}_{1}$ state but with $h_{j}=0$. The results for $\mathcal{T}_{2}$ and $\mathcal{T}_{3}$ states are delayed by one period.  }
    \label{fig:eeim4}
\end{figure}

\begin{figure}[H]
    \centering
    \includegraphics[width= \linewidth]{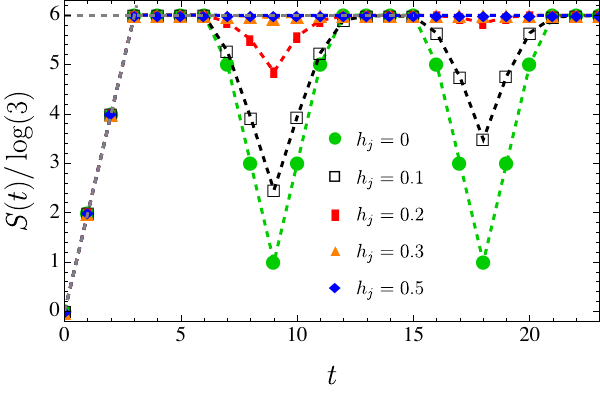}
    \caption{Time evolution of von  Neumann entanglement entropy for $q= 3$ model for a subsystem of $N=6$ qutrits in a system of $L=18$ qutrits KPM. For all the cases the initial state is of $\mathcal{T}_{1}$ type and we take different values for the longitudinal field.}\label{fig:eeimg3wib}
\end{figure}

\subsection{Generic states}
We now numerically study the fate of generic states, $ \ket{\psi_{\theta, \chi, \phi_{1},\phi_{2}}}$.
First we study the spread of entanglement at finite $L$ for generic states. As in previous section we study the chaotic and the integrable version of the model separately for their effects at finite system size. We observe from Fig. \eqref{fig:peeim3gsall} (a) corresponding to chaotic cases that the slope of the growth depends on the values of the parameters $\theta, \chi, \alpha$ and $\beta$. 
The saturation value, however, is universal and independent of the initial state,
showing its independence on the initial state. For the integrable case however, as shown in Fig. \eqref{fig:peeim3gsall} (b) we notice that the slope is no more universal and depends on the initial state chosen.

\begin{figure}[H]
    \centering
    \includegraphics[width= \linewidth]{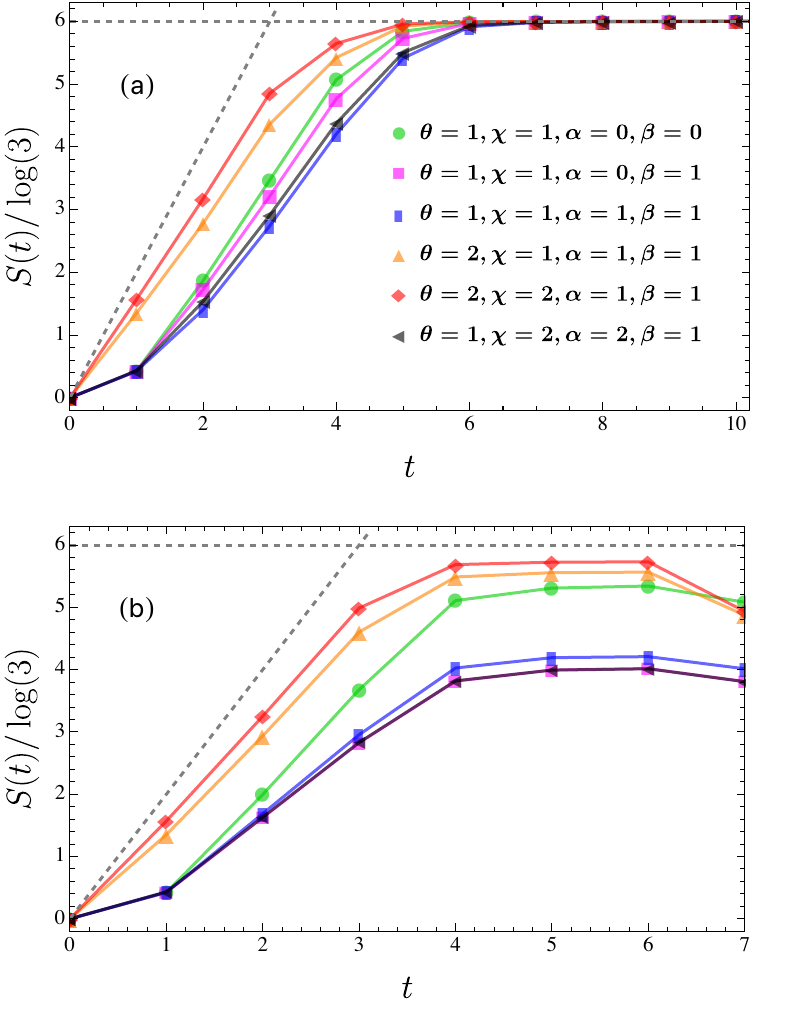}
    \caption{Time evolution of von  Neumann entanglement entropy for $q= 3$ model for a subsystem of $N=6$ qutrits in a system of $L=18$ qutrits KPM. The initial states are generic, $\ket{\psi_{\theta, \chi, \phi_{1},\phi_{2}}}$ ,and the corresponding parameter values are shown in the plots. (a)For $h_{j}=1$ (b)For $h_{j}=0$ which corresponds to integrable model.}\label{fig:peeim3gsall}
\end{figure}

Using the duality mapping we can also study the behavior of growth regime in the large $L$ and large $N$ limit. Following \cite{Bertini:2018fbz} we have following relation for the R\'enyi entropy in the thermodynamic limit
\begin{equation}\label{eq:renyithermo}
    \lim_{N\rightarrow \infty} \lim_{L\rightarrow \infty}S^{(n)}(t)=  \frac{2}{1-n} \log \left|\frac{\braket{L|\mathcal{P}|R}}{\braket{L|R}} \right|,
\end{equation}

where $\bra{L}$ and $\ket{R}$ are the left and right eigenvector respectively of $\mathbb{T}[h]$, corresponding to eigenvalue 1. 

\begin{figure}[H]
    \centering
    \includegraphics[width= \linewidth]{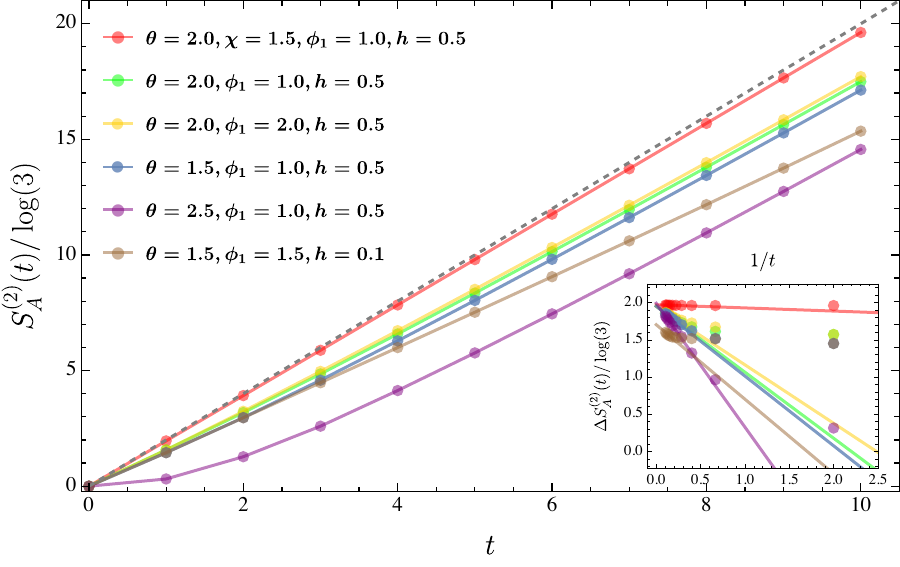}
    \caption{Time evolution of Renyi-2 entanglement entropy for KPM with $q= 3$ in thermodynamic limit. The initial states are generic, $\ket{\psi_{\theta, \chi, \phi_{1},\phi_{2}}}$ , and the corresponding parameter values are shown in the plots. We take $\theta= \chi$ and $\phi_{1}= \phi_{2}$ unless stated otherwise. The inset shows the instantaneous slope $\Delta S^{(\alpha)}_{A}(t-1/2)$[c.f. Eq. \eqref{eq:instslope}], as a function of $1/t$. The linear fit is obtained using values for $t\geq 5$.}\label{fig:eeim3thermoallp}
\end{figure}

Numerical evaluation of Eq. \eqref{eq:renyithermo} by further noting that $\ket{R}$ and $\bra{L}$ have a tensor product structure of the transfer matrix, thus we have
\begin{align}
 \ket{R}= \otimes_{\nu=1}^{n} \ket{V_{R}}, \quad \bra{L}= \otimes_{\nu=1}^{n} \ket{V_{L}}.  
\end{align}
where $\ket{V_{R}}, \ket{V_{L}} \in \mathcal{H}_{t} \otimes \mathcal{H}_{t}$. $\ket{V_{R}}$ and $\bra{V_{L}}$ can be obtained efficiently using the power method \cite{Bertini:2018fbz} and the R\'enyi entropy can then be written as 
\begin{equation}\label{eq:renyithermofinal}
    \lim_{N \rightarrow \infty} \lim_{L \rightarrow \infty} S^{(n)}_{A}(t)= \frac{2}{1-n} \log \left| \frac{\tr[(V^{\dagger}_{R}V_{L})^{n}]}{(\tr(V_{R}^{\dagger}V_{L}))^{n}}\right|.
\end{equation}
where $V_{L,R}$ are $3^{t} \times 3^{t}$ matrices corresponding to the vector $\bra{V_{L,R}}$ through vector operator correspondence. 

\begin{figure}[H]
    \centering
    \includegraphics[width= \linewidth]{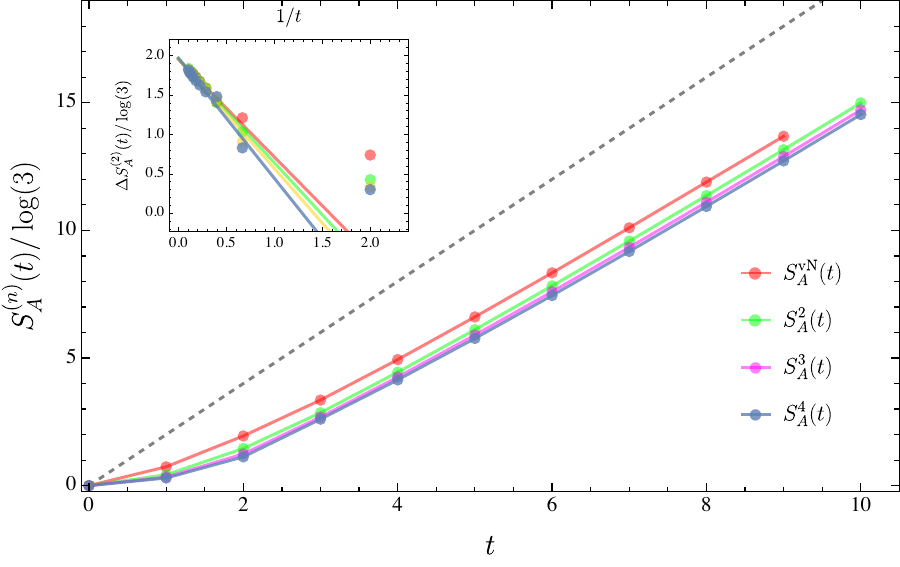}
    \caption{Time evolution of R\'enyi entropies, $\alpha=1,2,3,4$, for KPM with $q= 3$ in thermodynamic limit. The initial states are generic, with $\theta= \chi=\phi_{1}=\phi_{2}=1,h=0.5$. The inset shows the instantaneous slope $\Delta S^{(\alpha)}_{A}(t-1/2)$[c.f. Eq. \eqref{eq:instslope}], as a function of $1/t$. The linear fit is obtained using values for $t\geq 5$.}\label{fig:renyinthermoall}
\end{figure}

Using Eq. \eqref{eq:renyithermofinal}, for various generic states we obtain the result for R\'enyi-2, shown in Fig. \ref{fig:eeim3thermoallp}. We observe that for all the states the growth is approximately linear (excluding some small $t$ behavior) but the exact slope depends on the initial state chosen and field $h_{i}$. We however find that these deviations vanish at large $t$. To see this we obtain the instantaneous slope defined as follows
\begin{equation}\label{eq:instslope}
    \Delta S^{(\alpha)}_{A}(t-1/2)= S^{(\alpha)}_{A}(t)-S^{(\alpha)}_{A}(t-1).
\end{equation}
We observe that the slope becomes a linear function of $1/t$ at large times, see Fig. \ref{fig:eeim3thermoallp} inset. Extrapolating the results to $t=\infty$ we observe that for the chaotic case the results are consistent with 
\begin{equation}
   \lim _{t \rightarrow \infty } \frac{\Delta S^{(\alpha)}_{A}(t)}{\log(3)} \Big|_{h \neq 0}=2.
\end{equation}
On the other hand for the integrable case we find the results are consistent with 
\begin{equation}
   \lim _{t \rightarrow \infty } \frac{\Delta S^{(\alpha)}_{A}(t)}{\log(3)}\Big|_{h = 0} \leq 2.
\end{equation}

Similarly, in Fig. \ref{fig:renyinthermoall} we obtained the result for various R\'enyi entropies. It can be seen that for generic states the R\'enyi entropies have different slope, with different R\'enyi index. However, a careful analysis of instantaneous slope reveals that the differences vanish in the $t \rightarrow \infty$ limit. 
\section{Outlook}
In this work, we investigated the dynamics of entanglement in generalized kicked Ising chains with local Hilbert-space dimension $q>2$. We demonstrated that dual-unitary points persist for $q=3$ and $q=4$, allowing for exact analytical treatment of entanglement growth, while such points cease to exist for $q\ge5$. Our results therefore reveal that dual unitarity, although remarkably robust in low-dimensional systems, becomes increasingly constrained as the local Hilbert-space dimension grows.
This $q=4$ threshold is reminiscent of the well-known change in the two-dimensional $q$-state Potts/parafermion story: the ferromagnetic Potts transition is continuous for $q\leq 4$, while it becomes first order for $q>4$. The continuous low-$q$ transitions are described by the $\mathbb{Z}_q$ parafermion CFTs whereas for $q>4$ the generic Potts transition no longer realizes this critical behavior~\cite{FateevZamolodchikov1985}. It would be interesting to understand whether the obstruction to dual unitarity found here has any deeper relation to this familiar $q=4$ boundary.

A central outcome of our analysis is that the breakdown of dual unitarity for large $q$ originates from an obstruction in simultaneously equalizing all transition amplitudes generated by the kick operator. Physically, this reflects the increasing number of in-equivalent interference channels between different clock states in higher-dimensional local Hilbert spaces. From this perspective, our work suggests that exact dual-unitary dynamics may be a rather special property of low-dimensional Floquet systems.

 One important question is related to  the large $q$ limit. Our asymptotic analysis revealed an intimate relation between dual-unitarity conditions and Bessel functions. It would be interesting to understand whether this structure admits a semiclassical or hydrodynamic interpretation in the large local-dimension limit, potentially connecting dual-unitary dynamics to broader notions of quantum chaos and information spreading in many-body systems. It would also be interesting to explore more general Floquet circuits and clock models beyond the kicked Ising-type construction studied here. In particular, one may ask whether alternative gate structures, additional symmetries, or multi-parameter driving protocols could restore dual-unitary behavior even at higher local dimensions. Finally, our exact results for entanglement growth in non-qubit systems provide a useful benchmark for future studies of quantum information dynamics in higher-dimensional local Hilbert spaces. We hope that our work stimulates further investigations into the interplay between dual unitarity, scrambling, and the structure of local Hilbert spaces in quantum many-body systems.

\textbf{Acknowledgments:} Tanay is supported by JST CREST (Grant No. JPMJCR24I2). A part of numerical calculations have been performed using the computational facilities of the Yukawa Institute for Theoretical Physics.
HE is supported by JST CREST (Grant
No. JPMJCR24I3). 
TP acknowledges
support from the European Research Council (ERC) through the Advanced Grant QUEST (Grant Agreement No. 101096208) and the Slovenian Research and Innovation agency (ARIS) through the Program P1-0402.

\twocolumngrid
\bibliography{references}
\hbox{}\thispagestyle{empty}\newpage

\onecolumngrid
\raggedbottom
\appendix
\renewcommand{\theequation}{\Alph{section}\arabic{equation}}
\numberwithin{equation}{section}
\section{Particle time duality for $q=3$ Kicked Potts model}\label{append:ptduality}
In this section we focus dual-unitary in $q=3$, KPM. For this model we have the following $Z$ and $X$ operators with $\omega= e^{2 \pi i/3}$
\begin{equation}\label{eqn:z3x3}
Z=\left(\begin{array}{ccc}
1 & 0 & 0 \\
0 & \omega & 0 \\
0 & 0 & \omega^2
\end{array}\right), \quad X=\left(\begin{array}{ccc}
0 & 1 & 0 \\
0 & 0 & 1 \\
1 & 0 & 0
\end{array}\right) .
\end{equation}
and the property that $X^{3}= Z^{3}= 1$.

Furthermore, we notice that all the kicks factorizes so we only need to find what does the kick $i$ does to the basis states to calculate the relevant contribution to the partition function. It suffices then to find $e^{-i b (X + X^{\dagger})}$. It can be easily found to be 
\begin{equation}
 e^{-i b (X + X^{\dagger})} =    \left(
\begin{array}{ccc}
 \frac{1}{3} e^{-2 i b} \left(1+2 e^{3 i b}\right) & -\frac{1}{3} e^{-2 i b} \left(-1+e^{3 i b}\right) & -\frac{1}{3} e^{-2 i b} \left(-1+e^{3 i b}\right)  \\
 -\frac{1}{3} e^{-2 i b} \left(-1+e^{3 i b}\right) & \frac{1}{3} e^{-2 i b} \left(1+2 e^{3 i b}\right) & -\frac{1}{3} e^{-2 i b} \left(-1+e^{3 i b}\right) \\
 -\frac{1}{3} e^{-2 i b} \left(-1+e^{3 i b}\right) & -\frac{1}{3} e^{-2 i b} \left(-1+e^{3 i b}\right) & \frac{1}{3} e^{-2 i b} \left(1+2 e^{3 i b}\right) \\
\end{array}
\right)
\end{equation}

Since the basis for the above is the $Z$ basis, we readily get the contribution
\begin{align} \label{eqn:xkick}
    \braket{p|e^{-i b (X_{j} + X_{j}^{\dagger})}|p} &=  \frac{1}{3} e^{-2 i b} \left(1+2 e^{3 i b}\right) \equiv A, \nonumber \\
    \braket{p|e^{-i b (X_{j} + X_{j}^{\dagger})}|q}&=  \frac{1}{3} e^{-2 i b} \left(1- e^{3 i b}\right) \equiv B \, , p \neq q.
\end{align}
Consider following parameters
\begin{equation}\label{eqn:alphabeta}
    \alpha= -i \log(B),\quad \beta= -i (\log(A) - \log(B)).
\end{equation}
This allow us to write the contribution as a single exponential
\begin{align}
    \braket{p|e^{-i b (X_{j} + X_{j}^{\dagger})}|q} &=  e^{-i (\alpha + \beta \delta_{pq})}.
\end{align}
we can check if $p = q$ we get $A$ and $B$ otherwise. To be noticed here $p$ and $q$ are the time indices here. We make few more changes by noticing that we can write the $\delta_{pq}$ as 
\begin{align}
    \delta_{pq} &= \frac{1}{3}(1+ s_{p}s^{*}_{q} + (s_{p}s^{*}_{q})^{2}), \quad s_{i}= \omega^{i}.
\end{align}
We can thus write the contribution of the kick as 
\begin{align}
    \braket{p|e^{-i b (X_{j} + X_{j}^{\dagger})}|q} &=  e^{-i (\alpha + \beta \frac{1}{3}(1+ s_{p}s^{*}_{q} + (s_{p}s^{*}_{q})^{2}))}, \nonumber \\
  \braket{p|e^{-i b (X_{j} + X_{j}^{\dagger})}|q} &=  e^{-i (k_{0} + k_{1}( s_{p}s_{q}^{*} + s^{*}_{p}s_{q}))},
\end{align}
with $k_{0}=\alpha+ \beta/3 $ and $k_1= \beta/3$.

The above two relations make the duality in space time explicit by noting that if exchange $n \leftrightarrow \tau$ and $J \leftrightarrow k_{1}$. For unitarity in both the direction we require $J= k_{1}= \frac{\beta}{3}$ which can be found to be at $J = b = \frac{2 \pi}{9}, \frac{4 \pi}{9}$. To see this, we first note that 
\begin{equation}
 k_{1}= -\frac{i}{3} \log(A/B)= -\frac{i}{3} \log\left(\frac{1 + 2 e^{3 i b}}{1-e^{3 i b}} \right) .   
\end{equation}
Now we can demand $|1 + 2 e^{3 i b}|= |1-e^{3 i b}|$ (so that $k_1$ is real) which gives us $\Re(e^{- 3 i b})= -1/2$ this gives us $b= \frac{2 \pi}{9},\frac{4 \pi}{9}$. Putting this value back in $k_{1}$ we get 
\begin{align}
  k_{1}=   -\frac{i}{3} \log\left(\frac{1 + 2 e^{ 2 \pi i/3}}{1-e^{2\pi i /3}} \right) . 
\end{align}
The term inside the $\log$ can be evaluated to be $e^{2 \pi i /3}$ which thus give 
\begin{align}
    k_{1} = - \frac{i}{3} \log(e^{2 \pi i /3})= \frac{2\pi}{9}
\end{align}
 and similarly for $b = \frac{4 \pi}{9}$. This show that these are indeed self-dual points as shown in \cite{PhysRevB.105.144306,Claeys:2024tuy} using complex Hadamard matrices. 
\section{Particle time duality for $q=4$ Kicked Potts model}
Next, we briefly discuss the dual-unitary in $q=4$ Potts model. Following the strategy as in previous section we can determine the dual-unitary point to be $J = b = \frac{\pi}{4},\frac{3\pi}{4}$.
We can also use the connection of KPM to complex Hadamard matrix to obtain the dual-unitary points, following \cite{Claeys:2024tuy}. We first note that we have following $Z$ and $X$ matrices with $\omega= e^{2 \pi i/4}$.
\begin{align}
Z=\left(\begin{array}{cccc}
1 & 0 & 0 &0 \\
0 & \omega & 0 &0 \\
0 & 0 & \omega^2 &0 \\
0 & 0 & 0 & \omega^3
\end{array}\right), \quad
X=\left(\begin{array}{cccc}
0 & 1 & 0 & 0 \\
0 & 0 & 1 & 0 \\
0 & 0 & 0 & 1 \\
1 & 0 & 0 & 0
\end{array}\right) .
\end{align}
and the property that $X^{4}= Z^{4}= 1$.
This gives us the complex Hadamard matrices as 
\begin{equation}
H_{1}=     \left(
\begin{array}{cccc}
 -i & 1 & i & 1 \\
 1 & i & 1 & -i \\
 i & 1 & -i & 1 \\
 1 & -i & 1 & i \\
\end{array}
\right), \quad H_{2}=  \left(
\begin{array}{cccc}
 1 & -i & -1 & -i \\
 -i & 1 & -i & -1 \\
 -1 & -i & 1 & -i \\
 -i & -1 & -i & 1 \\
\end{array}
\right).
\end{equation}
These matrices can be shown to be equivalent to following complex Hadamard \cite{Tadej:2006ytn}
\begin{equation}
    F_4^{(1)}(a)=\left(\begin{array}{cccc}
1 & 1 & 1 & 1 \\
1 & ie^{i a} & -1 & -ie^{i a} \\
1 & -1 & 1 & -1 \\
1 & -i e^{i a} & -1 & i e^{i a}
\end{array}\right), \quad a \in[0, \pi).
\end{equation}

It is easy to show that 
\begin{align}
    H_{1} &= D_{1,1}. F_{4}^{(1)}(\pi/2) . P . D_{2,1}, \nonumber \\
    H_{2} &= D_{1,2}. F_{4}^{(1)}(\pi/2).D_{1,2}.
\end{align}
where we have 
\begin{align}
    D_{1,1} &= diag(-i,1,i,1), \nonumber \\
    D_{2,1} &= diag(1,i,-1,i), \nonumber \\
    P &= \left(
\begin{array}{cccc}
 1 & 0 & 0 & 0 \\
 0 & 0 & 0 & 1 \\
 0 & 0 & 1 & 0 \\
 0 & 1 & 0 & 0 \\
\end{array}
\right), \nonumber\\
D_{1,2}&=   diag(1,-i,-1,-i).
\end{align}

Also notice that unlike the case of $q=3$ when $H_{1}$ and $H_{2}$ were hermitian conjugate of each other,  we have two different matrices here. It is also to be noticed that the above matrices are obtained for the special point $J = b =  \frac{\pi}{4},\frac{3\pi}{4}$ which can be identified as the dual-unitary point.


\section{Properties of $\mathbb{T}[h]$}\label{sec:tprop}

The transfer matrix can be shown to have some simple properties. Specifically we can show that 

\begin{prop*}
Transfer matrix $\mathbb{T}[h]$ has following properties
\begin{itemize}
    \item[\text{(i)}] $|\lambda_{i}| \leq 3^n \max \left(\cos ^{2 n}\left(\frac{\theta}{2}\right),\, \sin ^{2 n}\left(\frac{\theta}{2}\right) \left(\frac{1}{2} (| \cos (\chi)| +1)\right)^n\right) $, $\forall \lambda_{i} \in \textrm{Spec}(\mathbb{T}[h])$.
     \item[\text{(ii)}] If $\lambda$ is an eigenvalues of $\mathbb{T}[h]$ such that $|\lambda|= \lambda_{max}$, then 
     \begin{itemize}
         \item[(a)] $\lambda$ has trivial Jordon block i.e. its geometric and algebraic multiplicities coincide. 
         \item[(b)] The associated left eigenvector $\bra{V_{L}}$ has the following properties
          \begin{align}
              \bra{V_{L}} \mathbb{H}_{\nu,1}[\theta,\chi] &= \lambda_{\mathrm{max}} \bra{V_{L}}, \label{eqn:hprop}\\
              \bra{V_{L}} \mathcal{P}^{(Z)}_{\nu,t}[\theta,\chi] &= \bra{V_{L}}, \label{eqn:pprop} \\
              \bra{V_{L}} \mathbb{U}^{(\nu)}_{\alpha,\beta}[h] &= e^{i \phi_{\nu}} \bra{V_{L}}. \label{eqn:uprop}
          \end{align}
     \end{itemize}
\end{itemize}
\end{prop*}
where $\textrm{Spec}(M)$ denotes the spectrum of  $M$.

\textit{Proof}-- Let $\bra{V}$ be vector so we have 
\begin{align}
    \braket{V| \mathbb{T}[h]\mathbb{T}^{\dagger}[h]|V} &= \braket{V| \mathbb{H}[\theta,\chi] .\mathcal{P}^{(Z)}_{\nu,t}. \mathbb{U}[h] | V}, \nonumber \\
    &\leq 3^n \max \left(\cos ^{2 n}\left(\frac{\theta}{2}\right),\, \sin ^{2 n}\left(\frac{\theta}{2}\right) \cos ^{2 n}\left(\frac{\chi}{2}\right),\, \sin ^{2 n}\left(\frac{\theta}{2}\right) \sin ^{2 n}\left(\frac{\chi}{2}\right)\right), \\
    &= 3^n \max \left(\cos ^{2 n}\left(\frac{\theta}{2}\right),\, \sin ^{2 n}\left(\frac{\theta}{2}\right) \left(\frac{1}{2} (| \cos (\chi)| +1)\right)^n\right).
\end{align}
Where the inequality follows from the fact that $\mathcal{P}^{(Z)}_{\nu,t}$ is a projector. Hence, its expectation value on a normalized state is less than or equal to 1.

If we now choose $\bra{V}$ to be the left eigenvector $\equiv \bra{V_{L}}$, of $\mathbb{T}[h]$ with eigenvalue $\lambda$ then we get the property 
\begin{equation}
    |\lambda| \leq 3^n \max \left(\cos ^{2 n}\left(\frac{\theta}{2}\right),\, \sin ^{2 n}\left(\frac{\theta}{2}\right) \left(\frac{1}{2} (| \cos (\chi)| +1)\right)^n\right)= \lambda_{\textrm{max}},
\end{equation}
which proves property (i). 

To prove property (ii a), we assume that the Jordon block of $\lambda$ is non trivial. Let $\bra{V_{L}}$ be the left eigenvector with corresponding eigenvalue $\lambda$ and $\bra{V'}$ be the first generalized eigenvector and since the Jordon block is non- trivial so we have 
\begin{equation}
    \bra{V'}\mathbb{T}[h]= \lambda \bra{V'} + a \bra{V}, \quad a \neq 0.
\end{equation}
This implies that we have 
\begin{equation}
    \bra{V'}\mathbb{T}[h]\mathbb{T}^{\dagger}[h] \ket{V'}= |\lambda_{max}|^{2} + |a|^{2},
\end{equation}
which is impossible as it contradicts property (i). This implies that Jordon block is trivial. 

The first and the second part of property (ii b) follows from fact that the \emph{equality} in property (i) is achieved when the expectation of projectors give 1. Finally since $\bra{V_{L}}$ is the left eigenvector of $\mathbb{T}[h]$ so we have $\bra{V_{L}} \mathbb{T}[h]= \lambda_{max}e^{i \phi} \bra{V_{L}}$. Now, using the first and the second part of property (ii b) the last part follows. 

We can also obtain few properties of transfer matrix using Eq. \eqref{eqn:trrhofull}. First we consider the case
\begin{equation}
    h_{i}= h, \theta_{i}= \theta, \chi_{i}= \chi, \alpha_{i}= \alpha, \beta_{i}= \beta. 
\end{equation}
Substituting $N=0$ in Eq. \eqref{eqn:trrhofull}, we obtain 
\begin{equation}
    \tr[(\mathbb{T}[h])^{L}]= \tr(\rho(t))=1.
\end{equation}
which follows from the fact that the original state is a pure state. This then implies that the eigenvalues of $\mathbb{T}[h]$ are all 0 except one which is 1. Also the corresponding Jordon block is one dimensional. thus we have following
\begin{prop*}
The spectrum of $\mathbb{T}[h]$ is fully determined as follows
\begin{enumerate}
    \item $\text{Spec}(\mathbb{T}[h])= \{0,1\}$.
    \item The geometric multiplicity of the eigenvalue 1 is one. 
\end{enumerate}
\end{prop*}
This means that maximum eigenvalue property can be saturated only when $\lambda_{\textrm{max}}=1$ which is obtained for 
\begin{equation}
\cos \left(\frac{\theta}{2}\right)= \sqrt{\frac{1}{3}}, \quad  \sin\left(\frac{\theta}{2}\right) \cos\left(\frac{\chi}{2}\right)= \sqrt{\frac{1}{3}}, \quad \sin\left(\frac{\theta}{2}\right) \sin\left(\frac{\chi}{2}\right)= \sqrt{\frac{1}{3}}.
\end{equation}
Using which we obtain one such state as
\begin{equation} \label{eqn:solvableval}
 \cos \left(\frac{\theta}{2}\right)= \sqrt{\frac{1}{3}} , \quad  \sin \left(\frac{\theta}{2}\right)= \sqrt{\frac{2}{3}}, \quad \chi = \frac{\pi}{2}.
\end{equation}
\section{Maximal eigenvalues and corresponding eigenvector}\label{sec:tevec}

We now determine the maximal eigenvalue of $\mathbb{T}[h]$, for specific values given by Eq. \eqref{eqn:solvableval}, and the associated eigenvector by searching for all eigenvector such that $\lambda_{\mathrm{max}}=1$. 

We first note that for special values given by Eq. \eqref{eqn:solvableval} we have
\begin{equation}
    \mathbb{H}_{\nu,1}[\theta,\chi] = \mathbb{1}\otimes \mathbb{1},
\end{equation}
hence Eq. \eqref{eqn:hprop} becomes trivial.

We then have the problem reduced to 
\begin{align}
    \bra{V_{L}} \mathcal{P}^{(Z)}_{\nu,t} &= \bra{V_{L}}, \\
\bra{V_{L}} \mathbb{U}^{(\nu)}_{\alpha,\beta}[h] &= e^{i \phi_{\nu}} \bra{V_{L}}, \quad \phi_{\nu} \in \mathbb{R} \quad \forall \nu.
\end{align}

To solve these equations we can use the following vector to operator map. For any operator $M$ we have the following
\begin{equation}
    \bra{M}= \sum_{jk} \braket{j | M |k} \bra{k} \otimes \bra{j}^{*}, 
\end{equation}
where $(\cdot)^*$ implies conjugation. It is also important for us to note $X$ and $Z$ in themselves do not form a group. that the full group algebra. However, if we consider operators $X^{a}Z^{b}$ for $a,b \in \{0,1,2\}$, then they span the algebra of group $\mathfrak{s l}(3, \mathbb{C})$.

Using this we can recast Eq. \eqref{eqn:pprop} and \eqref{eqn:uprop} as 
\begin{align}
    Z^{-1}_{\nu,t} V_{L}Z_{\nu,t} + Z_{\nu,t} V_{L}Z^{-1}_{\nu,t}& =  2 V_{L}\label{eqn:vlprop2a},  \\ 
    U_{\nu,\alpha,\beta}e^{-i h M_{\nu}^{z}}e^{-i J M_{\nu}^{x}} V_{L}&= e^{i \phi_{\nu}} V_{L}  U_{\nu,\alpha,\beta}e^{-i h M_{\nu}^{z}}e^{-i J M_{\nu}^{x}} \label{eqn:vlprop2b}.
\end{align}

Since $ Z_{\nu,t}$ is diagonal, then the first relation above implies following
\begin{equation}
    Z^{-1}_{\nu,t} V_{L}Z_{\nu,t} = Z_{\nu,t} V_{L}Z^{-1}_{\nu,t} = V_{L} \implies [V_{L}, Z_{\nu,t}]=[V_{L}, Z^{\dagger}_{\nu,t}]=0.
\end{equation}

Multiplying both sides of Eq. \eqref{eqn:vlprop2b} by $e^{i J M_{\nu}^{x}} e^{i h M_{\nu}^{z}} U^{\dagger}_{\nu,\alpha,\beta}$ we obtain 
\begin{equation}\label{eqn:vlprop2c}
    V_{L} e^{i J M_{\nu}^{x}} e^{i h M_{\nu}^{z}} U^{\dagger}_{\nu,\alpha,\beta} = e^{\phi_{\nu}} e^{i J M_{\nu}^{x}} e^{i h M_{\nu}^{z}} U^{\dagger}_{\nu,\alpha,\beta}V_{L}.
\end{equation}

Using Eq. \eqref{eqn:vlprop2a}, \eqref{eqn:vlprop2b} and \eqref{eqn:vlprop2c} we can see that 
\begin{equation}
V_{L} e^{i J M_{\nu}^{x}} e^{i h M_{\nu}^{z}} U^{\dagger}_{\nu,\alpha,\beta} Z_{\nu,t} U_{\nu,\alpha,\beta}e^{-i h M_{\nu}^{z}}e^{-i J M_{\nu}^{x}} = e^{i J M_{\nu}^{x}} e^{i h M_{\nu}^{z}} U^{\dagger}_{\nu,\alpha,\beta} Z_{\nu,t} U_{\nu,\alpha,\beta}e^{-i h M_{\nu}^{z}}e^{-i J M_{\nu}^{x}} V_{L}
\end{equation}
Using the relation $ e^{i J M_{\nu}^{x}} e^{i h M_{\nu}^{z}} U^{\dagger}_{\nu,\alpha,\beta} Z_{\nu,t} U_{\nu,\alpha,\beta}e^{-i h M_{\nu}^{z}}e^{-i J M_{\nu}^{x}}= e^{i J M_{\nu}^{x}} Z_{\nu,t} e^{-i J M_{\nu}^{x}}$ we can thus conclude 
\begin{align}
    [V_{L}, e^{i J M_{\nu}^{x}} Z_{\nu,t} e^{-i J M_{\nu}^{x}}] =0, \nonumber \\
  \implies  [V_{L}, e^{i J (X_{\nu,t}+X^{\dagger}_{\nu,t})} Z_{\nu,t} e^{-i J (X_{\nu,t}+X^{\dagger}_{\nu,t})}] =0. \label{eqn:vlprop2d}
\end{align}

Next, consider the relation $e^{- i J X}$. It can be evaluated to 
\begin{equation}
    e^{- i J X}= c_{0} \mathbb{1} + c_{1} X +c_{2}X^{\dagger},
\end{equation}
where $c_{i} = \sum_{m=0}^{\infty} \frac{(-i J)^{3m+i}}{(3 m +i)!}$. 
So we have 
\begin{align}\label{eqn:expx3}
    e^{- i J (X+ X^{\dagger})}&= (c_{0} \mathbb{1} + c_{1} X +c_{2}X^{\dagger} )(c_{0} \mathbb{1} + c_{1} X^{\dagger} +c_{2}X ), \nonumber \\
    &= (c_{0}^{2} + c_{1}^{2}+ c_{2}^{2}) \mathbb{I} + (c_{0}c_{1}+c_{1}c_{2}+c_{2}c_{3}) (X+ X^{\dagger}), \nonumber \\
    & = \tilde{C}_{0} \mathbb{I} + \tilde{C}_{1} (X+ X^{\dagger}),
\end{align}
where $\tilde{C}_{0}= \frac{1}{3} \left(2 e^{i J}+e^{-2 i J}\right)$ and $\tilde{C}_{1}= \frac{1}{3} e^{-2 i J} \left(1-e^{3 i J}\right)$ with $J = \frac{2 \pi}{9}.$

So we have 
\begin{align}
    e^{i J (X+ X^{\dagger})} Z e^{- i J (X+ X^{\dagger})} &= (\tilde{C}^{*}_{0} \mathbb{I} + \tilde{C}^{*}_{1} (X+ X^{\dagger}))Z (\tilde{C}_{0} \mathbb{I} + \tilde{C}_{1} (X+ X^{\dagger})), \nonumber \\
    &=  \omega^{2} ZX.
\end{align}

Combining with Eq. \eqref{eqn:vlprop2d} we obtain
\begin{equation}
[V_{L}, Z_{\nu,t}X_{\nu,t}] = 0. 
\end{equation}
Similarly repeating the calculation with $Z_{\nu,t}^{\dagger}$, we obtain 
\begin{equation}
    [V_{L}, Z^{\dagger}_{\nu,t}X_{\nu,t}] = 0. 
\end{equation}
To obtain $[V_{L},X^{\dagger}]$, we notice that 
\begin{align}
 [V_{L},Z_{\nu,t} X_{\nu,t}]= [V_{L},Z_{\nu,t}] X_{\nu,t} + Z_{\nu,t}^{\dagger} [V_{L}X_{\nu,t}] = 0, \nonumber \\
 \implies [V_{L}, X_{\nu,t}]=0.
\end{align}
 We thus obtain $[V_{L}, X^{\dagger}_{\nu,t}]=0$.
 
 This small exercise allows use to show that 
\begin{equation}\label{eqn:vlprop3}
    [V_{L}, X_{\nu,t}^{a}Z_{\nu,t}^{b}]=0 ; \quad a,b \in \{0,1,2\},
\end{equation}
and this is important as $X^{a}Z^{b} \quad \forall\, a,b \in \{0,1,2\}$ generates the algebra. 

We now proceed to use proof by induction to prove that if 
\begin{align}
  [V_{L}, X_{\nu,\mu}^{a}Z_{\nu,\mu}^{b}]=0 ; \quad a,b \in \{0,1,2\}   \nonumber \\
  \mu \in \{\bar{\mu}+1, \cdots, t\}
\end{align}
then 
\begin{equation}
    [V_{L}, X_{\nu,\bar{\mu}}^{a}Z_{\nu,\bar{\mu}}^{b}]= 0 \quad \forall\, a,b \in \{0,1,2\}
\end{equation}
The base of the induction is given by Eq. \eqref{eqn:vlprop3}. We thus only need to prove the inductive step.

First we note that $V_{L}$ commutes with 
\begin{equation}
    U_{\nu, \alpha,\beta} e^{-i h M_{\nu}^{Z}} e^{-i J M_{\nu}^{X}} X_{\nu,\mu}e^{ i J M_{\nu}^{X}} e^{-i h M_{\nu}^{Z}}  U^{\dagger}_{\nu, \alpha,\beta}= U_{\nu, \alpha,\beta} e^{-i h M_{\nu}^{Z}}  X_{\nu,\mu} e^{-i h M_{\nu}^{Z}}  U^{\dagger}_{\nu, \alpha,\beta},
\end{equation}
$\mu \in \{\bar{\mu}+1, \cdots, t\}$.
We now note that 
\begin{equation}
    e^{- i h M_{\nu}^{Z}} X_{\nu,\mu} e^{i h M_{\nu}^{Z}} = \mathcal{C}_0 X_{\nu,\mu}+ \mathcal{C}_1 X_{\nu,\mu}Z_{\nu,\mu} +\mathcal{C}_2X_{\nu,\mu}Z^{\dagger}_{\nu,\mu},
\end{equation}
for some non-zero constants $\mathcal{C}_{i}; i = 0,1,2$. Thus we have 
\begin{equation}
    U_{\nu, \alpha,\beta} e^{-i h M_{\nu}^{Z}} e^{-i J M_{\nu}^{X}} X_{\nu,\mu}e^{ i J M_{\nu}^{X}} e^{-i h M_{\nu}^{Z}}  U^{\dagger}_{\nu, \alpha,\beta}= U_{\nu, \alpha,\beta} (\mathcal{C}_0 X_{\nu,\mu}+ \mathcal{C}_1 X_{\nu,\mu}Z_{\nu,\mu} +\mathcal{C}_2 X_{\nu,\mu}Z^{\dagger}_{\nu,\mu}) e^{-i h M_{\nu}^{Z}}  U^{\dagger}_{\nu, \alpha,\beta}.
\end{equation}
We now note that $V_{L}$ commutes with the left hand side and thus 
\begin{align}
   V_{L} U_{\nu, \alpha,\beta} e^{-i h M_{\nu}^{Z}} e^{-i J M_{\nu}^{X}} X_{\nu,\mu}e^{ i J M_{\nu}^{X}} e^{-i h M_{\nu}^{Z}}  U^{\dagger}_{\nu, \alpha,\beta} &= U_{\nu, \alpha,\beta} e^{-i h M_{\nu}^{Z}} e^{-i J M_{\nu}^{X}} X_{\nu,\mu}e^{ i J M_{\nu}^{X}} e^{-i h M_{\nu}^{Z}}  U^{\dagger}_{\nu, \alpha,\beta} V_{L}, \nonumber \\
   V_{L} U_{\nu, \alpha,\beta} (\mathcal{C}_0 X_{\nu,\mu}+ \mathcal{C}_1 X_{\nu,\mu}Z_{\nu,\mu} +\mathcal{C}_2 X_{\nu,\mu}Z^{\dagger}_{\nu,\mu}) e^{-i h M_{\nu}^{Z}}  U^{\dagger}_{\nu, \alpha,\beta} &=  U_{\nu, \alpha,\beta} (\mathcal{C}_0 X_{\nu,\mu}+ \mathcal{C}_1 X_{\nu,\mu}Z_{\nu,\mu} +\mathcal{C}_2 X_{\nu,\mu}Z^{\dagger}_{\nu,\mu}) e^{-i h M_{\nu}^{Z}}  U^{\dagger}_{\nu, \alpha,\beta} V_{L}. 
\end{align}
Thus we have $[V_{L}, U_{\nu,\alpha,\beta} X_{\nu,\mu}U_{\nu,\alpha,\beta}]=0$. Now, focusing only on the term $U_{\nu,\alpha,\beta} X_{\nu,\mu}U_{\nu,\alpha,\beta}$ we have 
\begin{align}
    U_{\nu,\alpha,\beta} X_{\nu,\mu}U_{\nu,\alpha,\beta}= e^{-i J (Z_{\nu,\mu-1}Z^{\dagger}_{\nu,\mu} + Z^{\dagger}_{\nu,\mu-1}Z_{\nu,\mu} +Z_{\nu,\mu}Z^{\dagger}_{\nu,\mu+1}+Z^{\dagger}_{\nu,\mu}Z_{\nu,\mu+1})} X_{\nu,\mu} e^{i J (Z_{\nu,\mu-1}Z^{\dagger}_{\nu,\mu} + Z^{\dagger}_{\nu,\mu-1}Z_{\nu,\mu} +Z_{\nu,\mu}Z^{\dagger}_{\nu,\mu+1}+Z^{\dagger}_{\nu,\mu}Z_{\nu,\mu+1})}
\end{align}
We now note that 
\begin{align}
   e^{-i J (Z_{\nu,\mu-1}Z^{\dagger}_{\nu,\mu} + Z^{\dagger}_{\nu,\mu-1}Z_{\nu,\mu})} X_{\nu,\mu} e^{i J (Z_{\nu,\mu-1}Z^{\dagger}_{\nu,\mu} + Z^{\dagger}_{\nu,\mu-1}Z_{\nu,\mu})}  \propto Z^{\dagger}_{\mu-1}X_{\mu}, \\
   e^{-i J (Z_{\nu,\mu}Z^{\dagger}_{\nu,\mu+1} + Z^{\dagger}_{\nu,\mu}Z_{\nu,\mu+1})} X_{\nu,\mu} e^{i J (Z_{\nu,\mu}Z^{\dagger}_{\nu,\mu+1} + Z^{\dagger}_{\nu,\mu}Z_{\nu,\mu+1})}  \propto X_{\mu} Z^{\dagger}_{\mu+1},
\end{align}
which combined with $[V_{L}, U_{\nu,\alpha,\beta} X_{\nu,\mu}U_{\nu,\alpha,\beta}]=0$ requires 
\begin{equation}
    [V_{L}, Z_{\nu,\mu-1}]=0.
\end{equation}
Note that if we now substitute $\mu= \bar{\mu}+1$ then it requires that 
\begin{equation}
    [V_{L}, Z_{\nu,\bar{\mu}}]=0,
\end{equation}
which proves the induction. 

Following all previous reasoning once again we can now show that 
\begin{align}
  [V_{L}, X^{a}_{\nu,\mu}Z^{b}_{\nu,\mu}]=0  \quad a, b \in \{0,1,2\}, \nonumber \\
  \mu \in\{1, \cdots,t\}, \nu \in \{1,\cdots,n\}.
\end{align}
This implies that $V_{L}$ commutes with the entire algebra, $\mathfrak{s l}(3, \mathbb{C})$. Since the algebra is irreducible, Schur's lemma implies that $V_{L}$ is trivial, upto a multiplicative factor. Thus we 
\begin{equation}
    V_{L} = \mathbb{1}\quad \text{and} \quad \phi_{\nu}=0.
\end{equation}
Thus we can write 
\begin{equation}
    \bra{\mathbb{1}}= \frac{1}{3^{nt/2}} \sum_{\{s_{\nu,t}\}} \bra{\{s_{\nu,t} \}} \otimes \bra{\{s_{\nu,t} \}},
\end{equation}
where $\bra{s_{\nu,t}}$ are the computational basis states. 

\section{Spectrum at  finite $N$}\label{append:trfiniten}

We first note that we have following properties \cite{Bertini:2018fbz}
\begin{align}
    \prod_{\nu=1}^{n} O_{\nu} \otimes O_{\nu}^{*} \ket{\Psi}= \ket{\Psi}, \label{eqn:oprop1}\\
   \bra{\Psi} \prod_{\nu=1}^{n} O_{\nu} \otimes O_{\nu}^{*}= \bra{\Psi},  \label{eqn:oprop2}
\end{align}
where $O_{\nu}$ is a unitary operator that acts non-trivially only on the $\nu$-th copy of $\mathcal{H}_{t}$  in $\mathcal{H}_{nt}$ and $P \ket{\mathbb{1}}= \ket{\Psi}$.
\begin{claim}
    For state corresponding to values given by Eq. \eqref{eqn:solvableval} we have
\begin{align}
\langle\Psi|\left(\prod_{j=1}^N \mathbb{T}\left[h_j\right]\right)|\Psi\rangle =\langle\Psi| \prod_{\nu=1}^n\left[\prod_{\tau=0}^{\lfloor N / 2\rfloor-1}\left[\mathcal{P}_{\nu, t-\tau}^{(Z)} \mathcal{C}_{\nu, t-\tau}^{(X)}\right]\left[\mathcal{P}_{\nu, t-\lfloor N / 2\rfloor}^Z\right]^{\mod(N,2)}\right] \ket{\Psi},
\end{align}
$\mathbb{T} \equiv \mathbb{T}_{\theta,\chi,\alpha,\beta}$ evaluated at values given by Eq. \eqref{eqn:solvableval}.
\end{claim}
\emph{Proof}- We first recall the following definitions,
\begin{align}
     \mathcal{P}^{(Z)}_{\nu,\tau}&=  \frac{\mathbb{1} + Z_{\nu,\tau} \otimes Z_{\nu,\tau}^{-1} + Z_{\nu,\tau}^{2} \otimes Z_{\nu,\tau}^{-2}}{3}, \\
      \mathcal{C}_{\nu,\tau}^{(X)} &= \frac{1}{3}( \mathbb{1}+ X_{\nu,\tau}\otimes X_{\nu,\tau}+ X_{\nu,\tau}^{\dagger}\otimes X_{\nu,\tau}^{\dagger}).
\end{align}

Defining the following 
\begin{align}
    \mathbb{A}_{\nu,\tau}= e^{- i J (Z_{\nu,\tau}Z^{\dagger}_{\nu,\tau+} + Z_{\nu,\tau}^{\dagger}Z_{\nu,\tau+1})} \otimes e^{ i J (Z_{\nu,\tau}Z^{\dagger}_{\nu,\tau+} + Z_{\nu,\tau}^{\dagger}Z_{\nu,\tau+1})},
\end{align}
\begin{align}
    \mathbb{Z}_{\nu,\tau}= e^{- i h (Z_{\nu,\tau} + Z_{\nu,\tau}^{\dagger})} \otimes e^{i h(Z_{\nu,\tau} + Z_{\nu,\tau}^{\dagger})},
\end{align}

\begin{align}
    \mathbb{X}_{\nu,\tau}= e^{- i J (X_{\nu,\tau} + X_{\nu,\tau}^{\dagger})} \otimes e^{i J (X_{\nu,\tau} + X_{\nu,\tau}^{\dagger})},
\end{align}

\begin{align}
    \mathcal{U}^{\alpha,\beta}_{\nu,1}= \exp[-i \alpha P^{(1)}_{\nu,1}-i \beta P^{(2)}_{\nu,1}] \otimes \exp[i \alpha P^{(1)}_{\nu,1}+i \beta P^{(2)}_{\nu,1}].
\end{align}
We thus have 

\begin{align}
\langle\Psi|\left(\prod_{j=1}^N \mathbb{T}\left[h_j\right]\right)|\Psi\rangle = \braket{ \Psi|\prod_{\nu=1}^{n}\left[ \prod_{j=1}^{N} \mathcal{P}_{\nu,t}^{(Z)} \mathcal{U}_{\nu,1} \prod_{\tau=1}^{t-1}\mathbb{A}_{\nu,\tau} \prod_{\tau=1}^{t} \mathbb{Z}^{h_{j}}_{\nu,\tau} \prod_{\tau=1}^{t} \mathbb{X}_{\nu,\tau} \right]  |\Psi}.
\end{align}

Next, using Eq. \eqref{eqn:oprop1} and \eqref{eqn:oprop2} we have the following
\begin{align}
\bra{\Psi} \prod_{\nu=1}^{n}\mathbb{A}_{\nu,\tau} &= \bra{\Psi}, \\
\bra{\Psi} \prod_{\nu=1}^{n}\mathbb{Z}_{\nu,\tau} &= \bra{\Psi}, \\
\bra{\Psi} \prod_{\nu=1}^{n}\mathbb{X}_{\nu,\tau} &= \bra{\Psi}, \\
\ket{\Psi} \prod_{\nu=1}^{n}\mathbb{A}_{\nu,\tau} &= \ket{\Psi}, \\
\ket{\Psi} \prod_{\nu=1}^{n}\mathbb{Z}_{\nu,\tau} &= \ket{\Psi}, \\
\ket{\Psi} \prod_{\nu=1}^{n}\mathbb{X}_{\nu,\tau} &= \ket{\Psi}.
\end{align}
We also have the following commutation relation
\begin{align}
    \mathbb{A}_{\nu,\tau}\mathbb{Z}^{h}_{\nu,\tau'}&= \mathbb{Z}^{h}_{\nu,\tau'}\mathbb{A}_{\nu,\tau} \, \forall \, \tau,\tau', \\ 
    \mathbb{A}_{\nu,\tau}\mathbb{X}^{h}_{\nu,\tau'}&= \mathbb{X}^{h}_{\nu,\tau'}\mathbb{A}_{\nu,\tau} \, \, \tau'\neq \tau, \tau+1, \\ 
    \mathbb{X}_{\nu,\tau}\mathbb{Z}^{h}_{\nu,\tau'}&= \mathbb{Z}^{h}_{\nu,\tau'}\mathbb{X}_{\nu,\tau} \, \, \tau'\neq \tau, \\ 
    \mathbb{A}_{\nu,\tau}\mathcal{P}^{(Z)}_{\nu,\tau'}&= \mathcal{P}^{(Z)}_{\nu,\tau'}\mathbb{A}_{\nu,\tau} \, \forall \, \tau,\tau', \\ 
    \mathbb{Z}_{\nu,\tau}\mathcal{P}^{(Z)}_{\nu,\tau'}&= \mathcal{P}^{(Z)}_{\nu,\tau'}\mathbb{Z}_{\nu,\tau} \, \forall \, \tau,\tau' ,\\ 
    \mathbb{X}_{\nu,\tau}\mathcal{P}^{(Z)}_{\nu,\tau'}&= \mathcal{P}^{(Z)}_{\nu,\tau'}\mathbb{X}_{\nu,\tau} \, \forall \, \tau,\tau .
\end{align}

Making use of these properties we now obtain 
\begin{align}
 \langle\Psi|\left(\prod_{j=1}^N \mathbb{T}\left[h_j\right]\right)|\Psi\rangle  =  \braket{\Psi|\prod_{\nu=1}^{n} \left[ \prod_{j=1}^{N-1}  \mathbb{B}_{\nu,j}\right]|\Psi },
\end{align}

where 
\begin{align}\label{eqn:bdef1}
    \mathbb{B}_{\nu,j}= \mathcal{P}^{(Z)}_{\nu,t} \prod_{\tau=t-j+1}^{t-1} \mathbb{A}_{\nu,\tau} \prod_{\tau=t-j+1}^{t} \mathbb{\tilde{Z}}^{h_{j},\alpha,\beta}_{\nu,\tau} \prod_{\tau=t-j+2}^{t} \mathbb{X}_{\nu,\tau} \mathcal{P}^{(Z)}_{\nu,t}, \\
    \mathbb{\tilde{Z}}^{h_{j},\alpha_{j},\beta_{j}}_{\nu,\tau} = \mathbb{Z}^{h_{j}}_{\nu,\tau} ((1-\delta_{\tau,1})\mathbb{1} + \delta_{\tau,1} \mathcal{U}^{\alpha,\beta}_{\nu,1}).
\end{align}

We also note that 
\begin{align}
    \mathcal{P}^{(Z)}_{\nu,\tau} \mathbb{X}_{\nu,\tau}\mathcal{P}^{(Z)}_{\nu,\tau}&= \mathcal{P}^{(Z)}_{\nu,\tau} \left( \mathcal{C}^{X}_{\nu,\tau} + \frac{\omega}{3} ( \mathbb{1}\otimes X_{\nu,\tau} +\mathbb{1}\otimes X^{\dagger}_{\nu,\tau} + \frac{\omega^{2}}{3} ( X_{\nu,\tau}\otimes \mathbb{1} +X^{\dagger}_{\nu,\tau} \otimes \mathbb{1}) + \frac{1}{3}( X_{\nu,\tau} \otimes X^{\dagger}_{\nu,\tau} + X^{\dagger}_{\nu,\tau} \otimes X_{\nu,\tau})\right)\mathcal{P}^{(Z)}_{\nu,\tau} \nonumber\\
    &= \mathcal{P}^{(Z)}_{\nu,\tau} \mathcal{C}^{X}_{\nu,\tau}, \label{eqn:prop2}
 \\
 \mathcal{C}_{\nu,\tau}^{(X)}&= \frac{1}{3}( \mathbb{1}+ X_{\nu,\tau}\otimes X_{\nu,\tau}+ X_{\nu,\tau}^{\dagger}\otimes X_{\nu,\tau}^{\dagger}).  
\end{align}
where we use Eq. \eqref{eqn:expx3} and following properties
\begin{align}
 \left(\frac{\mathbb{1} + \omega Z_{\nu,\tau} \otimes Z^{\dagger}_{\nu,\tau} + \omega^{2} Z_{\nu,\tau}^{\dagger} \otimes Z_{\nu,\tau}}{3}\right)\mathcal{C}_{\nu,t}^{(X)} = 0, \\
 \left(\frac{\mathbb{1} + \omega^{2} Z_{\nu,\tau} \otimes Z^{\dagger}_{\nu,\tau} + \omega Z^{\dagger}_{\nu,\tau} \otimes Z_{\nu,\tau}}{3}\right) \mathcal{C}_{\nu,t}^{(X)} =0, \\
    \mathcal{C}_{\nu,\tau}^{X} \mathcal{P}_{\nu,\tau'}^{(Z)} = \mathcal{P}_{\nu,\tau'}^{(Z)}\mathcal{C}_{\nu,\tau}^{(X)}\quad \forall \quad \tau, \tau'.
\end{align}

So we can write Eq. \eqref{eqn:bdef1} as 
\begin{align}
    \mathbb{B}_{\nu,j}= \mathcal{P}^{(Z)}_{\nu,t}\mathbb{A}_{\nu,t-1}  \mathcal{C}_{\nu,t}^{(X)}\prod_{\tau=t-j+1}^{t-2} \mathbb{A}_{\nu,\tau} \prod_{\tau=t-j+2}^{t-1} \mathbb{\tilde{Z}}^{h_{j},\alpha,\beta}_{\nu,\tau} \prod_{\tau=t-j+1}^{t-1} \mathbb{X}_{\nu,\tau}.
\end{align}

We then have following lemma

\begin{lemma}\label{eqn:blemma}
    \begin{align}
    \mathbb{B}_{\nu,1} \cdots \mathbb{B}_{\nu,2n} &= \mathbb{A}_{\nu,t-1} \prod_{j=0}^{n-1}[\mathcal{P}^{(Z)}_{\nu,t-j}\mathcal{C}^{(X)}_{\nu,t-j}] \mathcal{P}^{(Z)}_{\nu,t-n} \mathbb{X}_{\nu,t-n} \prod^{2n-2}_{j=n} \left[ \prod^{t-2n+j}_{\tau=t-1-j} \mathbb{A}_{\nu,\tau} \prod_{\tau=t-j}^{t-2n+j} \mathbb{\tilde{Z}}^{h_{j},\alpha,\beta}_{\nu,\tau} \prod_{\tau=t-2n+j+1}^{t-1-j} \mathbb{X}_{\nu,\tau}  \right],\label{eqn:blemmaa} \\
    \mathbb{B}_{\nu,1} \cdots \mathbb{B}_{\nu,2n+1} &= \mathbb{A}_{\nu,t-1} \prod_{j=0}^{n}[\mathcal{P}^{(Z)}_{\nu,t-j}\mathcal{C}^{(X)}_{\nu,t-j}] \mathcal{P}^{(Z)}_{\nu,t-n} \mathbb{X}_{\nu,t-n} \prod^{2n-1}_{j=n} \left[ \prod^{t-2n-1+j}_{\tau=t-1-j} \mathbb{A}_{\nu,\tau} \prod_{\tau=t-j}^{t-2n+j-1} \mathbb{\tilde{Z}}^{h_{j},\alpha,\beta}_{\nu,\tau} \prod_{\tau=t-2n+j}^{t-1-j} \mathbb{X}_{\nu,\tau}  \right].\label{eqn:blemmab}
    \end{align}
\end{lemma}
Using this we finally obtain 
\begin{align}
\langle\Psi|\left(\prod_{j=1}^N \mathbb{T}\left[h_j\right]\right)|\Psi\rangle = \braket{ \Psi|\prod_{\nu=1}^{n}\left[ \prod_{j=1}^{\lfloor N/2 \rfloor-1} [\mathcal{P}_{\nu,t-j}^{(Z)} \mathcal{C}_{\nu,t-j}^{(X)}  \right] [ \mathcal{P}_{\nu,t- \lfloor N/2 \rfloor}^{(Z)}]^{\mathrm{mod}(N,2)}  |\Psi}.
\end{align}

\section{Proof of Lemma 1}
We prove the lemma using induction. 

\emph{Proof}- The base of the induction is given by $n=1$, which is given by 
\begin{equation}\label{eqn:blemman1}
    \mathbb{B}_{\nu,1}\mathbb{B}_{\nu,2}= \mathcal{P}^{(Z)}_{\nu,t}\mathbb{A}_{\nu,t-1}  \mathcal{C}_{\nu,t}^{(X)} \mathcal{P}^{(Z)}_{\nu,t}\mathbb{A}_{\nu,t-1}  \mathcal{C}_{\nu,t}^{(X)} \mathbb{X}_{\nu,t-1}= \mathbb{A}_{\nu,t-1}\mathcal{P}_{\nu,t}^{(Z)}\mathcal{C}^{(X)}_{\nu,t}\mathcal{P}_{\nu,t-1}^{(Z)}\mathbb{X}_{\nu,t-1}.
\end{equation}
It can be proven by expanding the $\mathbb{A}_{\nu,\tau}$ and the using the following properties
\begin{align}
  \left(\frac{\mathbb{1} + \omega X_{\nu,\tau} \otimes X_{\nu,\tau} + \omega^{2} X_{\nu,\tau}^{\dagger} \otimes X_{\nu,\tau}^{\dagger}}{3}\right)\mathcal{C}_{\nu,t}^{(X)} = 0, \\
 \left(\frac{\mathbb{1} + \omega^{2} X_{\nu,\tau} \otimes X_{\nu,\tau} + \omega X_{\nu,\tau}^{\dagger} \otimes X_{\nu,\tau}^{\dagger}}{3}\right) \mathcal{C}_{\nu,t}^{(X)} =0,\\
 \mathcal{P}^{(Z)}_{\nu,t} (Z^{\nu,\tau} \otimes Z^{\nu,\tau} -\mathbb{1})= 0.
\end{align}
which can be used to obtain
\begin{equation}\label{eqn:prop1}
\mathcal{P}_{\nu,\tau}^{(Z)}\mathcal{C}^{(X)}_{\nu,\tau}\mathbb{A}_{\nu,\tau-1}\mathcal{P}^{(Z)}_{\nu,\tau}= \mathcal{P}^{(Z)}_{\nu,\tau}\mathcal{C}^{(X)}_{\nu,\tau}\mathcal{P}^{(Z)}_{\nu,\tau-1}.
\end{equation}
Eq. \eqref{eqn:blemman1} agrees with Eq. \eqref{eqn:blemmaa} for $n=1$.

Next we have 
\begin{align}
    \mathbb{B}_{\nu,1}\mathbb{B}_{\nu,2}\mathbb{B}_{\nu,3}&= \mathcal{P}^{(Z)}_{\nu,t}\mathbb{A}_{\nu,t-1}  \mathcal{C}_{\nu,t}^{(X)} \mathcal{P}^{(Z)}_{\nu,t}\mathbb{A}_{\nu,t-1}  \mathcal{C}_{\nu,t}^{(X)} \mathbb{X}_{\nu,t-1} \mathcal{P}^{(Z)}_{\nu,t}\mathbb{A}_{\nu,t-1}\mathcal{C}^{(X)}_{\nu,t}\mathbb{A}_{\nu,t-2} \mathbb{Z}^{h_{3},\alpha_{3},\beta_{3}}_{\nu,\tau-1}\mathbb{X}_{\nu,t-2}\mathbb{X}_{\nu,t-1},\\
    &= \mathbb{A}_{\nu,t-1}\mathcal{P}^{(Z)}_{\nu,t} \mathcal{C}^{(X)}_{\nu,t}\mathcal{P}^{(Z)}_{\nu,t-1} \mathcal{C}^{(X)}_{\nu,t-1}\mathbb{A}_{\nu,t-2}\mathbb{Z}^{h_{3},\alpha_{3},\beta_{3}}_{\nu,\tau-1}\mathbb{X}_{\nu,t-2}\mathbb{X}_{\nu,t-1},\\
    &= \mathbb{A}_{\nu,t-1}\mathcal{P}^{(Z)}_{\nu,t} \mathcal{C}^{(X)}_{\nu,t}\mathcal{P}^{(Z)}_{\nu,t-1} \mathcal{C}^{(X)}_{\nu,t-1}\mathbb{A}_{\nu,t-2}\mathbb{X}_{\nu,t-2}\mathbb{X}_{\nu,t-1},
\end{align}
which agrees with Eq. \eqref{eqn:blemmab} with $n=1$. In the last step we use the fact that 
\begin{equation}
\mathcal{P}^{(Z)}_{\nu,\tau}\mathbb{Z}^{h_{3},\alpha_{3},\beta_{3}}_{\nu,\tau-1}= \mathcal{P}^{(Z)}_{\nu,\tau}.
\end{equation} 
To obtain it we use find the expression for $e^{-i h( Z+ Z^{\dagger})}$, analogous to Eq. \eqref{eqn:expx3} with $\tilde{C}_{0}= \frac{1}{3} \left(2 e^{i h}+e^{-2 i h}\right)$ and $\tilde{C}_{1}= \frac{1}{3} e^{-2 i h} \left(1-e^{3 i h}\right)$. We then evaluate the obtained expression and use the following properties of the coefficients
\begin{equation}
    |\tilde{C}_{0}^{2}|= \frac{1}{9}(5+4\cos(3h)), \quad  |\tilde{C}_{1}^{2}|= \frac{4}{9}\sin \left(\frac{3h}{2}\right)^{2},\quad |\tilde{C}_{0}^{2}|+ 2|\tilde{C}_{1}^{2}|=1,\quad \tilde{C}^{*}_{0}\tilde{C}_{1}+ \tilde{C}_{0}\tilde{C}^{*}_{1}+ |\tilde{C}_{0}|^{2}=0.
\end{equation}

Hence the lemma \ref{eqn:blemma} holds for $n=1$. Next, we prove the base of the induction assuming Eq. \eqref{eqn:blemmaa} and show that Eq. \eqref{eqn:blemmab} holds. 

We now evaluate 
    \begin{align}
    \mathbb{B}_{\nu,1} \cdots \mathbb{B}_{\nu,2n}\mathbb{B}_{\nu,2n+1} &= \mathbb{A}_{\nu,t-1} \prod_{j=0}^{n-1}[\mathcal{P}^{(Z)}_{\nu,t-j}\mathcal{C}^{(X)}_{\nu,t-j}] \mathcal{P}^{(Z)}_{\nu,t-n} \mathbb{X}_{\nu,t-n} \prod^{2n-2}_{j=n} \left[ \prod^{t-2n+j}_{\tau=t-1-j} \mathbb{A}_{\nu,\tau} \prod_{\tau=t-j}^{t-2n+j} \mathbb{\tilde{Z}}^{h_{j},\alpha,\beta}_{\nu,\tau} \prod_{\tau=t-2n+j+1}^{t-1-j} \mathbb{X}_{\nu,\tau}  \right] \nonumber \\
 &\times   \mathcal{P}^{(Z)}_{\nu,t}\mathbb{A}_{\nu,t-1}  \mathcal{C}_{\nu,t}^{(X)}\prod_{\tau=t-2n+1}^{t-2} \mathbb{A}_{\nu,\tau} \prod_{\tau=t-2n+2}^{t-1} \mathbb{\tilde{Z}}^{h_{j},\alpha,\beta}_{\nu,\tau} \prod_{\tau=t-2n+1}^{t-1} \mathbb{X}_{\nu,\tau}.
    \end{align}
We now note that $\mathcal{C}^{(X)}_{\nu,t}$ commutes with all the terms in the first line, to bring it closer to $\mathcal{P}_{\nu,t}^{(Z)}$ and then we use Eq. \eqref{eqn:prop1} to obtain
 \begin{align}
    \mathbb{B}_{\nu,1} \cdots \mathbb{B}_{\nu,2n}\mathbb{B}_{\nu,2n+1} &= \mathbb{A}_{\nu,t-1} \prod_{j=0}^{n-1}[\mathcal{P}^{(Z)}_{\nu,t-j}\mathcal{C}^{(X)}_{\nu,t-j}] \mathcal{P}^{(Z)}_{\nu,t-n} \mathbb{X}_{\nu,t-n} \prod^{2n-2}_{j=n} \left[ \prod^{t-2n+j}_{\tau=t-1-j} \mathbb{A}_{\nu,\tau} \prod_{\tau=t-j}^{t-2n+j} \mathbb{\tilde{Z}}^{h_{j},\alpha,\beta}_{\nu,\tau} \prod_{\tau=t-2n+j+1}^{t-1-j} \mathbb{X}_{\nu,\tau}  \right] \nonumber \\
 &\times   \mathcal{P}^{(Z)}_{\nu,t-1}\prod_{\tau=t-2n+1}^{t-2} \mathbb{A}_{\nu,\tau} \prod_{\tau=t-2n+2}^{t-1} \mathbb{\tilde{Z}}^{h_{j},\alpha,\beta}_{\nu,\tau} \prod_{\tau=t-2n+1}^{t-1} \mathbb{X}_{\nu,\tau}, \nonumber \\
 &= \mathbb{A}_{\nu,t-1} \prod_{j=0}^{n-1}[\mathcal{P}^{(Z)}_{\nu,t-j}\mathcal{C}^{(X)}_{\nu,t-j}] \mathcal{P}^{(Z)}_{\nu,t-n} \mathbb{X}_{\nu,t-n}\mathcal{P}^{(Z)}_{\nu,t-n} \prod^{2n-2}_{j=n} \left[ \prod^{t-2n+j}_{\tau=t-1-j} \mathbb{A}_{\nu,\tau} \prod_{\tau=t-j}^{t-2n+j} \mathbb{\tilde{Z}}^{h_{j},\alpha,\beta}_{\nu,\tau} \prod_{\tau=t-2n+j+1}^{t-1-j} \mathbb{X}_{\nu,\tau}  \right] \nonumber \\
 &\times \prod_{\tau=t-2n+1}^{t-2} \mathbb{A}_{\nu,\tau} \prod_{\tau=t-2n+2}^{t-1} \mathbb{\tilde{Z}}^{h_{j},\alpha,\beta}_{\nu,\tau} \prod_{\tau=t-2n+1}^{t-1} \mathbb{X}_{\nu,\tau},
    \end{align}

where we repeatedly used Eq. \eqref{eqn:prop1} and \eqref{eqn:prop2} to obtain the last line. 

Finally we obtain,

 \begin{align}
    \mathbb{B}_{\nu,1} \cdots \mathbb{B}_{\nu,2n}\mathbb{B}_{\nu,2n+1} 
 &= \mathbb{A}_{\nu,t-1} \prod_{j=0}^{n-1}[\mathcal{P}^{(Z)}_{\nu,t-j}\mathcal{C}^{(X)}_{\nu,t-j}] \mathcal{P}^{(Z)}_{\nu,t-n} \mathcal{P}^{(Z)}_{\nu,t-n} \prod^{2n-2}_{j=n} \left[ \prod^{t-2n+j}_{\tau=t-1-j} \mathbb{A}_{\nu,\tau} \prod_{\tau=t-j}^{t-2n+j} \mathbb{\tilde{Z}}^{h_{j},\alpha,\beta}_{\nu,\tau} \prod_{\tau=t-2n+j+1}^{t-1-j} \mathbb{X}_{\nu,\tau}  \right] \nonumber \\
 &\times \prod_{\tau=t-2n+1}^{t-2} \mathbb{A}_{\nu,\tau} \prod_{\tau=t-2n+2}^{t-1} \mathbb{\tilde{Z}}^{h_{j},\alpha,\beta}_{\nu,\tau} \prod_{\tau=t-2n+1}^{t-1} \mathbb{X}_{\nu,\tau}, \nonumber \\
 &= \mathbb{A}_{\nu,t-1} \prod_{j=0}^{n}[\mathcal{P}^{(Z)}_{\nu,t-j}\mathcal{C}^{(X)}_{\nu,t-j}] \mathcal{P}^{(Z)}_{\nu,t-n} \mathbb{C}^{(X)}_{\nu,t-n} \prod^{2n-1}_{j=n} \left[ \prod^{t-2n+j}_{\tau=t-1-j} \mathbb{A}_{\nu,\tau} \prod_{\tau=t-j}^{t-2n+j} \mathbb{\tilde{Z}}^{h_{j},\alpha,\beta}_{\nu,\tau} \prod_{\tau=t-2n+j+1}^{t-1-j} \mathbb{X}_{\nu,\tau}  \right], 
 \end{align}
which is exactly Eq. \eqref{eqn:blemmab}.

Similarly, we can start from Eq. \eqref{eqn:blemmab} and show that Eq. \eqref{eqn:blemmaa} holds for $n+1$. This proves the induction and hence concludes the proof. 
\end{document}